\documentclass[11pt,reqno,a4paper]{amsart}
\usepackage[utf8]{inputenc}
\usepackage[T1]{fontenc}
\usepackage{amsbsy}
\usepackage{amsfonts}
\usepackage{amsmath}
\usepackage{amssymb}
\usepackage{amsthm}
\usepackage{bm}
\usepackage{dsfont}
\usepackage{hyperref}
\usepackage{mathrsfs}
\usepackage{upgreek}
\usepackage{slashed}

\newcommand{\1}{\mathds{1}}

\newcommand{\C}{\mathbb{C}}

\newcommand{\R}{\mathbb{R}}

\newcommand{\bA}{{\bm{A}}}

\newcommand{\bF}{{\bm{F}}}

\newcommand{\boF}{\mathcal{F}}

\newcommand{\boI}{\mathcal{I}}

\newcommand{\boL}{\mathcal{L}}

\newcommand{\bsalpha}{\bm{\alpha}}
\newcommand{\bsbeta}{\bm{\beta}}

\newcommand{\bssigma}{\bm{\sigma}}

\newcommand{\gS}{\mathfrak{S}}

\newcommand{\eps}{\varepsilon}
\renewcommand{\epsilon}{\varepsilon}
\newcommand\ii{{\ensuremath {\infty}}}

\newcommand{\norm}[1]{ \left| \! \left| #1 \right| \! \right| }
\DeclareMathOperator{\curl}{{\rm curl}}
\renewcommand{\div}{\mathop{\mathrm{div}}\nolimits}
\renewcommand{\Im}{\mathop{\mathrm{Im}}\nolimits}
\renewcommand{\Re}{\mathop{\mathrm{Re}}\nolimits}

\DeclareMathOperator{\tr}{{\rm tr}}

\newtheorem{theorem}{Theorem}

\newtheorem{cor}[theorem]{Corollary}
\newtheorem{lemma}[theorem]{Lemma}
\newtheorem{prop}[theorem]{Proposition}

\theoremstyle{definition}
\newtheorem*{merci}{Acknowledgements}
\newtheorem{remark}[theorem]{Remark}

\begin{document}

\title[Derivation of the magnetic Euler-Heisenberg Energy]{Derivation of the magnetic Euler-Heisenberg Energy}

\author[P. Gravejat]{Philippe GRAVEJAT}
\address{Universit\'e de Cergy-Pontoise, Laboratoire de Math\'ematiques (UMR CNRS 8088), F-95302 Cergy-Pontoise Cedex, France.}
\email{philippe.gravejat@u-cergy.fr}

\author[M. Lewin]{Mathieu LEWIN}
\address{CNRS \& Universit\'e Paris-Dauphine, PSL Research University, CEREMADE (UMR CNRS 7534), 75016 Paris, France.}
\email{Mathieu.Lewin@math.cnrs.fr}

\author[\'E. S\'er\'e]{\'Eric S\'ER\'E}
\address{Universit\'e Paris-Dauphine, PSL Research University, CEREMADE (UMR CNRS 7534), 75016 Paris, France.}
\email{sere@ceremade.dauphine.fr}

\date{\today}

\begin{abstract}
In quantum field theory, the vacuum is a fluctuating medium which behaves as a nonlinear polarizable material. In this article, we perform the first rigorous derivation of the magnetic Euler-Heisenberg effective energy, a nonlinear functional that describes the effective fluctuations of the quantum vacuum in a classical magnetic field. 

We start from a classical magnetic field in interaction with a quantized Dirac field in its ground state, and we study a limit in which the classical magnetic field is slowly varying. After a change of scales, this is equivalent to a semi-classical limit $\hbar\to0$, with a strong magnetic field of order $1/\hbar$. In this regime, we prove that the energy of Dirac's polarized vacuum converges to the Euler-Heisenberg functional. The model has ultraviolet divergences, which we regularize using the Pauli-Villars method. We also discuss how to remove the regularization of the Euler-Heisenberg effective Lagrangian, using charge renormalization, perturbatively to any order of the coupling constant.
\end{abstract}

\maketitle

\tableofcontents

\section{Introduction}

In quantum field theory, the vacuum is a fluctuating medium which behaves as a nonlinear polarizable material~\cite{Dirac7,EulKoc-35,Weissko1}. 
A convenient way of describing these effects is to use an \emph{effective action}. In Quantum Electrodynamics (QED), this method corresponds to integrating out the electronic degrees of freedom in the full QED functional integral~\cite[Chap.~33]{Schwartz-14}. The effective action is a function of a \emph{classical} electromagnetic field treated as an external one. In the case of a \emph{constant} electromagnetic field, the effective action has a rather simple explicit expression. This effective \emph{Euler-Heisenberg Lagrangian} has been used to make spectacular predictions~\cite{Schwartz-14,DitGie-00}. For instance, the birefringence of the vacuum that was predicted by this theory has only been confirmed recently~\cite{MigTesGonTavTurZanWu-17,Fan-etal-17}.

For time-independent fields in the Coulomb gauge, the effective Lagrangian action has been rigorously defined in~\cite{GrHaLeS1}. It takes the form
\begin{equation}
\label{eq:form-eff-Lagrangian}
\begin{split}
\boL( \bA) := & - \boF_{\rm vac}(e  \bA) + e \int_{\R^3} \big( j_{\rm ext}(x) \cdot A(x) - \rho_{\rm ext}(x) V(x) \big)\,dx\\
& + \frac{1}{8 \pi} \int_{\R^3} \big( |E(x)|^2 - |B(x)|^2 \big)\,dx
\end{split}
\end{equation}
and gives rise to nonlinear and nonlocal corrections to the classical linear Maxwell equations
\begin{equation}
\label{eq:Maxwell}
\begin{cases}
- \Delta V = 4 \pi e \big( \rho_{\rm vac}(e \bA) + \rho_{\rm ext} \big),\\
- \Delta A = 4 \pi e \big( j_{\rm vac}(e  \bA) + j_{\rm ext} \big),\\
\div A = \div  j_{\rm ext}=0.
\end{cases}
\end{equation}
Here $e$ is the elementary charge of an electron, $\bA := (V, A)$ is a classical, $\R^4$-valued, electromagnetic potential, with corresponding field 
$$\bF=(E,B)=(-\nabla V,\curl A),$$
and $\rho_{\rm ext}$ and $j_{\rm ext}$ are given external charge and current densities. The vacuum charge $\rho_{\rm vac}(e  \bA)$ and current $j_{\rm vac}(e  \bA)$ densities are nonlinearly induced by the electron-positron vacuum according to the formulae
$$e \rho_{\rm vac}(e  \bA) := \frac{\partial}{\partial V} \boF_{\rm vac}(e  \bA) , \quad {\rm and} \quad e j_{\rm vac}(e  \bA) := - \frac{\partial}{\partial A} \boF_{\rm vac}(e  \bA),$$
where $\boF_{\rm vac}(e  \bA)$ is the energy of a quantized Dirac field in the external four-potential $ \bA$, assumed to be in its ground state.

The nonlinear vacuum correction terms $ \rho_{\rm vac}(e  \bA)$ and $j_{\rm vac}(e  \bA)$ are negligible in normal conditions and they only start to play a significant role for very large electromagnetic fields. The critical value at which this starts to happen is called the \emph{Schwinger limit} and it is several orders of magnitude above what can now be produced in the laboratory. The detection of these nonlinear effects is nevertheless a very active area of experimental research~\cite{Burkeetal-97,MouTajBul-06}. In addition, these terms are known to play an important role in neutron stars with extremely high magnetic fields (magnetars)~\cite{BarHar-01,MarkBroSte-03,DenSve-03}. The observation of this effect has been announced very recently~\cite{MigTesGonTavTurZanWu-17}. The equations~\eqref{eq:Maxwell} do not seem to have attracted much attention on the mathematical side.

The exact vacuum energy $\boF_{\rm vac}(e  \bA)$ is a very complicated nonlocal functional of $ \bA$ which, in addition, has well-known divergences. Its precise definition will be recalled in Section~\ref{sub:PV} below.
A simple and useful approximation, often used in the Physics literature, consists in replacing the complicated functional $\boF_{\rm vac}(e  \bA)$ by a superposition of \emph{local} independent problems, that is,
\begin{equation}
 \boF_{\rm vac}(e  \bA)\simeq \int_{\R^3}f_{\rm vac}\big(eE(x),eB(x)\big)\,dx
 \label{eq:local_approx}
\end{equation}
where $f_{\rm vac}\big(eE,eB\big)$ is the energy per unit volume, found for an electromagnetic field $\bF=(E,B)$ which is constant everywhere in space (hence for a linear electromagnetic potential $\bA$). Our purpose in this article is to justify the approximation~\eqref{eq:local_approx} in a limit where the electromagnetic field $\bF$ is slowly varying.

The (appropriately renormalized) function $f_{\rm vac}\big(eE,eB\big)$ has been computed by Euler  and Heisenberg in a famous article~\cite{HeisEul1} and it is given by the expression
\begin{multline}
\label{def:F-vac-EH}
f_{\rm vac}(eE,eB)=
\frac{1}{8 \pi^2}  \int_0^\infty \frac{e^{- s m^2}}{s^3} \, \Bigg( \frac{e^2 s^2}{3} \, \big( |E|^2 - |B|^2 \big) - 1\\
 + e^2 s^2 \, \big( E \cdot B \big) \, \frac{\Re \cosh \big( e s \big( |B|^2 - |E|^2 + i E \cdot B \big)^\frac{1}{2} \big)}{\Im \cosh \big( e s \big( |B|^2 - |E|^2 + i E \cdot B \big)^\frac{1}{2} \big)} \Bigg)\, ds.
\end{multline}
Here $m$ is the mass of the electron and we work in a system of units such that the reduced Planck constant $\hbar$ and the speed of light $c$ are equal to $1$.
An alternative computation of the Euler-Heisenberg formula~\eqref{def:F-vac-EH} has been discussed later by Weisskopf~\cite{Weissko1} and Schwinger~\cite{Schwing2}. For weak fields the leading correction to the usual Maxwell Lagrangian is given by
$$f_{\rm vac}(eE,eB)=-\frac{e^4}{360\pi^2m^4}\Big((|E|^2-|B|^2)^2+7(E\cdot B)^2\Big)+o(|E|^4+|B|^4).$$
This expression has a form similar to the first-order Born-Infeld theory~\cite{Kiessling-04a,Kiessling-04b} and it has played an important role in the understanding of nonlinear effects on the propagation and dispersion of light~\cite{KarpNeu1,Minguzzi-56,KenPla-63,KleNig-64}. 

It is well-known that the Euler-Heisenberg vacuum energy in formula~\eqref{def:F-vac-EH} has to be handled with care, since the integrand may have poles on the real line. The proper definition requires to replace $s$ by $s+i\eta$ and to take the limit $\eta\to0$~\cite{Schwing2,JenGieValLamWen-02}. The function $f_{\rm vac}(eE,eB)$ defined in this way may have an exponentially small non-zero imaginary part, which is interpreted as the electron-positron pair production rate and corresponds to the instability of the vacuum (see~\cite[Paragraph 7.3]{GreiRei0} and~\cite{DitGie-00,JenGieValLamWen-02,Dunne-05,Dunne-12}). However, under the additional constraints that $E\cdot B=0$ and $|E|<|B|$, then there is no pole, the integral in~\eqref{def:F-vac-EH} converges absolutely and the vacuum energy $f_{\rm vac}(eE,eB)$ is real, as expected. In particular, this is the case for a purely magnetic field, $E\equiv0$, that will be the object of our paper.

The purely magnetic case has been particularly discussed in the Physics literature~\cite{Adler-71,Constantinescu-72,TsaErb-74,TsaErb-75,MelSto-76}. The corresponding energy is independent of the direction of the magnetic field $B$ and simplifies to
\begin{equation}
\label{def:F-vac-mEH}
f_{\rm vac}(0,eB)= \frac{1}{8 \pi^2}  \int_0^\infty \frac{e^{- s m^2}}{s^3} \, \bigg( e s |B| \coth \big( e s |B| \big) - 1 - \frac{e^2 s^2 |B|^2}{3} \bigg)\,ds.
\end{equation}
This function is concave-decreasing and negative, which corresponds to the usual picture that, after renormalization, the vacuum polarization \emph{enhances} an external field instead of screening it as one would have naturally expected at first sight. As an illustration, the function behaves as
\begin{equation}
f_{\rm vac}(0,eB)\sim\begin{cases}
\displaystyle-\frac{e^4}{360\pi^2m^4}|B|^4&\text{for $|B|\to0$},\\[0.3cm]
\displaystyle-\frac{e^2|B|^2}{24\pi^2}\log\left(\frac{e|B|}{m^2}\right)&\text{for $|B|\to\infty$}.
\end{cases}
\label{eq:behavior_f} 
\end{equation}
Note that the energy diverges faster than $|B|^2$ at infinity. The total energy
$$\frac1{8\pi}\int_{\R^3}|B(x)|^2\,dx+\int_{\R^3} f_{\rm vac}\big(0,e|B(x)|\big)\,dx$$
is unbounded from below. The implications of this instability for strong magnetic fields were studied by Haba in~\cite{Haba-82,Haba-84}. The instability of the model is due to charge renormalization. 

The function $|B|\mapsto f_{\rm vac}(0,eB)$ has a \emph{non-convergent Taylor series} in powers of $B$. In fact, the series is 
$$\frac{m^4}{8\pi^2}\sum_{n\geq2}\frac{\mathcal{B}_{2n}}{2n(2n-1)(2n-2)}\left(\frac{2e|B|}{m^2}\right)^{2n}$$
and the Bernoulli numbers 
$$\mathcal{B}_{2n}=(-1)^{n+1}\frac{(2n)!}{(2\pi)^{2n}}\zeta(2n)$$
(with $\zeta$ the Riemann zeta function) diverge extremely fast. This phenomenon has played an important historical role in quantum field theory, where perturbative arguments are often used. It can nevertheless be shown that $f_{\rm vac}(0,eB)$ is the Borel sum of its Taylor expansion~\cite{ChaOle-77}.

Our main goal in this paper is to provide a rigorous derivation of the purely magnetic Euler-Heisenberg vacuum energy in~\eqref{def:F-vac-mEH}, starting from the vacuum energy $\boF_{\rm vac}$ of a quantized Dirac field, in the regime where $B$ varies slowly in space. To be more precise, we assume that the magnetic field takes the form $B(\varepsilon x)$ with a fixed smooth function $B$ (which corresponds to the strong magnetic potential $A_\varepsilon(x)=\varepsilon^{-1}A(\varepsilon x)$) and then look at the limit $\varepsilon\to0$. The expectation is that the vacuum energy is locally given by the Euler-Heisenberg formula, leading to the limit
\begin{equation}
\boF_{\rm vac}\big(0,e A_\varepsilon\big)\underset{\varepsilon\to0}\sim \int_{\R^3}f_{\rm vac}(0,eB(\varepsilon x))\,dx=\varepsilon^{-3}\int_{\R^3}f_{\rm vac}(0,eB(x))\,dx.
\label{eq:expected_result}
\end{equation}

After a change of scale, the problem coincides with a semiclassical limit $\hbar=\varepsilon\to0$ with a very strong magnetic field of order $B\sim1/\hbar$. This magnetic field strength is critical: any $B\ll 1/\hbar$ would disappear in the leading order of the semi-classical limit. Here $B$ does contribute, and deriving the exact form of the energy is a very delicate matter. 
This regime of strong magnetic fields has been the object of several recent studies in different situations~\cite{LieSoYn2, FuGuPeY1, Sobolev-95, Yngvaso1, Erdos-97, ErdoSol1, ErdoSol2, ErdoSol3, Fournais-01, Fournais-01b, Fournais-02, HelKorRayNgo-15}, none of them covering the continuous spectrum of the Dirac operator, to our knowledge. A result similar to~\eqref{eq:expected_result} has recently been obtained in the simpler case of a scalar field in~\cite{LamLew-15b}, which is a more conventional semiclassical limit.

In addition to the criticality of the magnetic field strength, we have to face the other difficulty that the vacuum energy $\boF_{\rm vac}$ has divergences that must be regularized. Different regularization schemes are possible. In this article we use the Pauli-Villars method as in~\cite{GrHaLeS1} and obtain a result similar to~\eqref{eq:expected_result}, with a Pauli-Villars-regulated function $f_{\rm vac}^{\rm PV}(0,eB)$. An advantage of the regularized model is that it is \emph{stable}, contrary to the original Euler-Heisenberg $f_{\rm vac}(0,eB)$. After charge renormalization, the function $f_{\rm vac}^{\rm PV}(0,eB)$ coincides with $f_{\rm vac}(0,eB)$ up to an exponentially small error, but the model becomes unstable.

In the next section we properly define the Dirac vacuum energy $\boF_{\rm vac}$ using the Pauli-Villars scheme, before we are able to state our main results.

\section{Main results}
\subsection{The Pauli-Villars-regulated vacuum energy}
\label{sub:PV}

We start by recalling the definition of the Pauli-Villars-regulated vacuum energy~\cite{PaulVil1} which was rigorously studied in~\cite{GrHaLeS1} (see also~\cite{GrHaLeS2,Lewin-6ECM}).

The vacuum energy $\boF_{\rm vac}(e \bA)$ in~\eqref{eq:form-eff-Lagrangian} is given by the \emph{formal} expression
\begin{equation}
\label{def:form-boF}
\boF_{\rm vac}(e  \bA) := - \frac{1}{2} \tr \big| D_{m, e  \bA} \big|.
\end{equation}
In this formula, 
\begin{equation}
\label{def:DA}
D_{m, e  \bA} := \bsalpha \cdot \big( - i \, \nabla - e A \big) + e V + m \, \bsbeta,
\end{equation}
is the Dirac operator for one electron in the classical electromagnetic field $\bA=(V,A_1,A_2,A_3)$~\cite{Thaller0,EstLeSe1}, a self-adjoint operator acting on $L^2(\R^3,\C^4)$. The four Dirac matrices $\bsalpha = (\bsalpha_1, \bsalpha_2 , \bsalpha_3)$ and $\bsbeta$ are given by
$$\bsalpha_k := \begin{pmatrix} 0 & \bssigma_k \\ \bssigma_k & 0 \end{pmatrix}, \quad {\rm and} \quad \bsbeta := \begin{pmatrix} I_2 & 0 \\ 0 & - I_2 \end{pmatrix},$$
with Pauli matrices $\bssigma_1$, $\bssigma_2$ and $\bssigma_3$ equal to
$$\bssigma_1 := \begin{pmatrix} 0 & 1 \\ 1 & 0 \end{pmatrix},\quad \bssigma_2 := \begin{pmatrix} 0 & - i \\ i & 0 \end{pmatrix}, \quad {\rm and} \quad \bssigma_3 := \begin{pmatrix} 1 & 0 \\ 0 & - 1 \end{pmatrix}.$$
The Dirac matrices satisfy the following anti-commutation relations
\begin{equation}
\label{eq:anti-commut}
\bsalpha_j \, \bsalpha_k + \bsalpha_k \, \bsalpha_j = 2 \delta_{j, k}\, I_4, \quad \bsalpha_j \, \bsbeta + \bsbeta \, \bsalpha_j = 0, \quad {\rm and} \quad \bsbeta^2 = I_4.
\end{equation}
The trace in~\eqref{def:form-boF} means that all the negative energy states are filled by virtual electrons, according to Dirac's picture~\cite{Dirac1,Dirac2,Dirac3,Dirac6,Dirac7}. The operator
$$- \frac{1}{2} \big| D_{m, e  \bA} \big|=D_{m, e  \bA}\frac{\1(D_{m, e  \bA}\leq0)-\1(D_{m, e  \bA}\geq0)}{2}$$
in the trace arises from the constraint that the system must be charge-conjugation invariant. We refer to~\cite{HaiLeSo1,GrHaLeS1} for detailed explanations.

Of course, the trace in~\eqref{def:form-boF} is infinite. In order to give a proper meaning to $\boF_{\rm vac}(e  \bA)$, we start by subtracting an infinite constant, namely the free vacuum energy corresponding to $ \bA\equiv0$ given by
$$\boF_{\rm vac}(0) = - \frac{1}{2} \tr \big| D_{m, 0} \big|.$$
We therefore consider the relative vacuum energy 
\begin{equation}
\label{def:rel-boF}
\boF_{\rm vac}(e \bA)-\boF_{\rm vac}(0)=\frac{1}{2} \tr \Big( \big| D_{m, 0} \big| - \big| D_{m, e  \bA} \big| \Big).
\end{equation}
Removing an (infinite) constant does not change the variational problem in which we are interested, as well as the resulting equations.

For most electromagnetic potentials $ \bA$, however, the relative vacuum energy in~\eqref{def:rel-boF} is not yet well defined due to ultraviolet divergences. This additional difficulty can be overcome by applying a suitable regularization. Various methods are employed in the literature, among which are the famous dimensional regularization~\cite{tHooVel1} and the lattice regularization~\cite{Wilson1}. Here we are going to use the Pauli-Villars method~\cite{PaulVil1} which was studied in~\cite{GrHaLeS1}. This technique consists in adding fictitious particles with high masses $m_j\gg1$ in the model. These particles have no physical significance, but their introduction provides an ultraviolet regularization, which is sufficient to rigorously define the vacuum energy. More precisely, we consider the so-called Pauli-Villars-regulated vacuum energy given by the formula
\begin{equation}
\label{def:F-vac-PV}
\boxed{\boF_{\rm vac}^{\rm PV}(e  \bA) := \frac{1}{2} \tr \sum_{j = 0}^2 c_j \, \Big( \big| D_{m_j, 0} \big| - \big| D_{m_j, e  \bA} \big| \Big).}
\end{equation}
In this expression, the coefficient $c_0$ and the mass $m_0$ are respectively equal to $1$ and $m$. The corresponding term is exactly the relative vacuum energy in~\eqref{def:rel-boF}.
The ultraviolet divergences are removed if the coefficients $c_1$ and $c_2$ satisfy the Pauli-Villars conditions
\begin{equation}
\label{eq:cond-PV}
\sum_{j = 0}^2 c_j = \sum_{j = 0}^2 c_j \, m_j^2 = 0,
\end{equation}
which amounts to choosing them as 
$$c_1 = - \frac{m_2^2 - m_0^2}{m_2^2 - m_1^2}, \quad {\rm and} \quad c_2 = \frac{m_1^2 - m_0^2}{m_2^2 - m_1^2}.$$
In the following we always assume that $m=m_0<m_1<m_2$ for simplicity, hence $c_1<0$ and $c_2>0$.

We now describe some known mathematical properties of $\boF_{\rm vac}^{\rm PV}$. Throughout the article we work in the electromagnetic field energy space, namely, we assume that $E,B\in L^2(\R^3)$.
This amounts to assuming that 
\begin{equation*}
V\in \dot{H}^1(\R^3) := \Big\{ V \in  L^6(\R^3, \R) \ : \ \| \nabla V \|_{L^2}^2 < \infty \Big\}
\end{equation*}
and
\begin{equation*}
A\in\dot{H}_{\div}^1(\R^3) := \Big\{ A \in  L^6(\R^3, \R^3) \ : \ \div A = 0\ \text{ and }\\ \| \curl A \|_{L^2}^2 < \infty \Big\}.
\end{equation*}
The precise result proved in~\cite{GrHaLeS1} is the following.

\begin{theorem}[Definition of $\boF_{\rm vac}^{\rm PV}$ in energy space~\cite{GrHaLeS1}]
\label{thm:def-L-PV}
Assume that the coefficients $c_0 = 1$, $c_1$ and $c_2$, and the masses $0<m = m_0 < m_1 < m_2$ satisfy the Pauli-Villars conditions~\eqref{eq:cond-PV}. 

\smallskip

\noindent $(i)$ When $V \in \dot{H}^1(\R^3)$ and $A\in \dot H^1_{\rm div}(\R^3)$, the operator 
\begin{equation}
T_{ \bA}:=\sum_{j = 0}^2 c_j \, \Big( \big| D_{m_j, 0} \big| - \big| D_{m_j,   \bA} \big| \Big)
\label{eq:def_Te}
\end{equation}
is compact.

\smallskip

\noindent $(ii)$ When $ V \in L^1(\R^3) \cap \dot{H}^1(\R^3)$ and $ A \in L^1(\R^3) \cap \dot{H}_{\rm div}^1(\R^3)$, we have $\tr\left|\tr_{\C^4}T_{ \bA}\right|<\infty$, hence  
$$\boF_{\rm vac}^{\rm PV}( \bA) := \frac{1}{2} \tr \left(\tr_{\C^4}T_{ \bA}\right)$$
is well-defined. Moreover, $\boF_{\rm vac}^{\rm PV}$ can be uniquely extended to a continuous mapping on $\dot{H}^1(\R^3) \times \dot{H}_{\div}^1(\R^3)$.

\smallskip

\noindent $(iii)$ When $V = 0$, and $A\in L^1(\R^3) \cap \dot{H}^1_{\rm div}(\R^3)$, the operator $T_{0,A}$ is trace-class on $L^2(\R^3, \C^4)$, that is, 
$\tr\left|T_{0,A}\right|<\infty$, hence $\boF^{\rm PV}_{\rm vac}(0,A)=\tr(T_{0,A})/2$.
\end{theorem}

The result says that the vacuum energy $\boF_{\rm vac}^{\rm PV}(  \bA)$ is well defined by a trace (possibly first doing the $\C^4$ trace), when $A$ and $V$ are smooth and integrable, and that it has a unique continuous extension to the energy space, still denoted by $\boF_{\rm vac}^{\rm PV}(  \bA)$. In the case $V\neq0$, we do not believe the operator $T_\bA$ to be trace-class without taking first the $\C^4$-trace. But the terms in $T_\bA$ which are (possibly) not trace-class do not contribute to the final value of $\boF_{\rm vac}^{\rm PV}( \bA)$, due to gauge-invariance.

With $\boF_{\rm vac}^{\rm PV}$ at hand, it is now possible to define the associated Lagrangian action, in external charge and current densities,
\begin{equation}
\label{eq:def-eff-Lagrangian}
\begin{split}
\boL^{\rm PV}( \bA) := & - \boF^{\rm PV}_{\rm vac}(e  \bA) + e \int_{\R^3} \big( j_{\rm ext}(x) \cdot A(x) - \rho_{\rm ext}(x) V(x) \big)\,dx\\
& + \frac{1}{8 \pi} \int_{\R^3} \big( |\nabla V(x)|^2 - |\curl A(x)|^2 \big)\,dx.
\end{split}
\end{equation}
provided that $\div j_{\rm ext}=0$ and
$$\rho_{\rm ext}\ast\frac{1}{|x|}\in \dot{H}^1(\R^3),\qquad  j_{\rm ext}\ast\frac{1}{|x|}\in \dot{H}^1_{\rm div}(\R^3).$$
In~\cite{GrHaLeS1} we have constructed the electromagnetic field $\bF = (E, B)$, in presence of time-independent, weak enough, external charge density $\rho_{\rm ext}$ and current density $j_{\rm ext}$, as a local min-max critical point of $\boL^{\rm PV}$. The corresponding four-potential $ \bA = (V, A)$ satisfies the nonlinear Maxwell equations in~\eqref{eq:Maxwell}. We refer to~\cite{GrHaLeS1, Lewin-6ECM, GrHaLeS2} for more details.

\subsection{Derivation of the Euler-Heisenberg vacuum energy}

We now come to our rigorous derivation of the Euler-Heisenberg vacuum energy in~\eqref{def:F-vac-mEH} starting from the Pauli-Villars-regulated vacuum energy. 

As announced we restrict our attention to purely magnetic fields by setting
$$V \equiv 0.$$
We next consider a scaled magnetic field of the form $B_\varepsilon(x) = B(\varepsilon x)$, with a given $B=\curl A\in L^2(\R^3)$, and $A\in \dot{H}^1_{\rm div}(\R^3)$. Our main result requires a bit more regularity.

\begin{theorem}[Derivation of the Euler-Heisenberg vacuum energy]
\label{thm:lim-eps}
Let $B\in C^0(\R^3,\R^3)$ be such that $\div B=0$ and 
\begin{equation}
\label{eq:cond-B}
B\in L^1(\R^3)\cap L^\ii(\R^3),\qquad \nabla B\in L^1(\R^3)\cap L^6(\R^3),
\end{equation}
and let $A$ be the associated magnetic potential in $\dot{H}^1_{\rm div}(\R^3)$. Set finally $A_\varepsilon(x)=\varepsilon^{-1} A(\varepsilon x)$.
Then, we have 
\begin{equation}
\label{eq:conv-eps}
\varepsilon^3 \boF_{\rm vac}^{\rm PV}(0,e A_\varepsilon)
=\int_{\R^3}f_{\rm vac}^{\rm PV}\big(e|B(x)|\big)\,dx+O(\varepsilon)
\end{equation}
where
\begin{equation}
\label{def:boF-PV-EH}
\begin{split}
f_{\rm vac}^{\rm PV}(b) := \frac{1}{8 \pi^2}  \int_0^\infty \bigg( \sum_{j = 0}^2 c_j \, e^{- s m_j^2} \bigg) \Big( s b \coth \big( s b \big) - 1 \Big)\frac{ds}{s^3}
\end{split}
\end{equation}
is the Pauli-Villars-regulated Euler-Heisenberg vacuum energy.
\end{theorem}

The theorem provides a limit in the same form as in~\eqref{eq:expected_result}, except that the effective energy $f_{\rm vac}^{\rm PV}(b)$ still depends on the regularization parameters $c_1,c_2,m_1,m_2$. In the next section we discuss the link with the original Euler-Heisenberg energy $f_{\rm vac}$. We remark that the function $f_{\rm vac}^{\rm PV}$ is non-negative and behaves as
\begin{equation}
f_{\rm vac}^{\rm PV}(eB)\sim\begin{cases}
\displaystyle \frac{e^2|B|^2}{24\pi^2}\sum_{j = 0}^2 c_j \, \log m_j^{-2}&\text{for $|B|\to0$},\\[0.3cm]
\displaystyle\frac{e|B|}{8\pi^2}\sum_{j=0}^2c_jm_j^2\log m_j^2&\text{for $|B|\to\infty$},
\end{cases}
\label{eq:behavior_f_PV} 
\end{equation}
where 
$$\sum_{j = 0}^2 c_j \, \log m_j^{-2}\geq0\quad\text{and}\quad \sum_{j=0}^2c_jm_j^2\log m_j^2\geq0,$$
as we show in Lemma~\ref{lem:positivity} below.
Note that the field energy is of the same order $\epsilon^{-3}$ as the vacuum in this regime,
$$\int_{\R^3}|B(\varepsilon x)|^2\,dx=\epsilon^{-3}\int_{\R^3}|B(x)|^2\,dx.$$
Therefore, in the limit $\epsilon\to0$ we get the effective local Lagrangian 
$$-\epsilon^{-3}\left(\frac1{8\pi}\int_{\R^3}|B(x)|^2\,dx+\int_{\R^3} f^{\rm PV}_{\rm vac}(e|B(x)|)\,dx\right).$$
The two terms have the same sign, hence the corresponding model is \emph{stable}, contrary to the renormalized functional based upon the Euler-Heisenberg function $f_{\rm vac}$.

\begin{remark}\it 
The technical conditions on $B$ in~\eqref{eq:cond-B} are certainly not optimal. It would be interesting to prove~\eqref{eq:conv-eps} with $O(\epsilon)$ replaced by $o(1)$, under the sole assumption that $\int_{\R^3} f^{\rm PV}_{\rm vac}(|B(x)|)\,dx<\ii$. Due to~\eqref{eq:behavior_f_PV}, this condition is equivalent to
$$\int_{\R^3}\frac{|B(x)|^2}{1+|B(x)|}\,dx<\infty.$$
\end{remark}

\subsection{Charge renormalization \& the Euler-Heisenberg energy}
The function $f_{\rm vac}^{\rm PV}$ defined in~\eqref{def:boF-PV-EH} and obtained in the limit $\epsilon\to0$ is different from the Euler-Heisenberg energy $f_{\rm vac}$. In particular, $f_{\rm vac}^{\rm PV}$ is positive convex increasing, whereas $f_{\rm vac}$ is negative concave decreasing. The reason is that $f_{\rm vac}^{\rm PV}$ does \emph{not} include the quadratic term $(sb)^2/3$ that is present in the formula~\eqref{def:F-vac-EH} of $f_{\rm vac}$. Indeed, integrating by parts using the first Pauli-Villars condition in~\eqref{eq:cond-PV}, we find
\begin{align}
 f_{\rm vac}^{\rm PV}(b)&=\frac{1}{8 \pi^2}  \int_0^\infty \bigg( \sum_{j = 0}^2 c_j \, e^{- s m_j^2} \bigg) \Big( s b \coth \big( s b \big) - 1 -\frac{(sb)^2}{3}\Big)\frac{ds}{s^3}\nonumber\\
 &\qquad +\frac{b^2}{24 \pi^2}  \int_0^\infty \bigg( \sum_{j = 0}^2 c_j \, e^{- s m_j^2} \bigg) \frac{ds}{s}\nonumber\\
&=f_{\rm vac}(b)+c_1\,\frac{m_1^4}{m^4}\,f_{\rm vac}\left(\frac{m^2}{m_1^2}b\right)+c_2\,\frac{m_2^4}{m^4}\,f_{\rm vac}\left(\frac{m^2}{m_2^2}b\right)+\frac{\log(\Lambda)}{24\pi^2}b^2
\label{eq:relation_PV_EH}
 \end{align}
where 
$$f_{\rm vac}(b)=\frac{1}{8 \pi^2}  \int_0^\infty  e^{- m^2s} \Big( s b \coth \big( s b \big) - 1 -\frac{(sb)^2}{3}\Big)\frac{ds}{s^3}$$
is the original Euler-Heisenberg energy with mass $m$, and where the averaged ultraviolet cut-off $\Lambda>1$ is defined by 
\begin{equation}
\label{def:Lambda}
\log \Lambda := - \sum_{j = 0}^2 c_j \, \log m_j>0.
\end{equation}
The second and third terms in~\eqref{eq:relation_PV_EH} are small since
$$c_j\,\frac{m_j^4}{m^4}\,f_{\rm vac}\left(\frac{m^2}{m_j^2}b\right)=O\left(\frac{|c_j|}{m_j^4}b^4\right),\qquad j=1,2.$$
In particular, we can take $m_1,m_2\to\infty$ while keeping $c_1,c_2$ bounded. However, the last term in~\eqref{eq:relation_PV_EH} is logarithmically-divergent. 

The correct way to remove the divergent term is to include it in the definition of the charge, a procedure called \emph{charge renormalization}. Here we explain this procedure on the real function $b\mapsto f_{\rm vac}^{\rm PV}(eb)$. By superposition, the same holds for the local function
$$B\mapsto \int_{\R^3}f_{\rm vac}^{\rm PV}\big(e|B(x)|\big)\,dx$$
obtained in the semi-classical limit. Following~\cite{HainSie1,HaiLeSe2,GraLeSe1,GraLeSe2}, we define the physical charge $e_{\rm ph}$ by
\begin{equation}
\label{def:e-ph}
e_{\rm ph}^2 = \frac{e^2}{1 + \frac{2 e^2}{3 \pi} \log \Lambda},
\end{equation}
and the physical magnetic field $b_{\rm ph}$ by the relation
$$e\,b=e_{\rm ph}\,b_{\rm ph}.$$
It is convenient to introduce the renormalization parameter
\begin{equation}
Z_3:=\frac{1}{1 + \frac{2 e^2}{3 \pi} \log \Lambda}\in(0,1)
\label{eq:def_Z_3}
\end{equation}
in such a way that $e_{\rm ph}=\sqrt{Z_3}\,e$ and $b_{\rm ph}=\sqrt{Z_3}\,b$. 

We would like to take $\Lambda\to\infty$, which however imposes that
$$e_{\rm ph}\leq \sqrt{\frac{3\pi}{2\log\Lambda}}\to0.$$
It is not possible to remove the ultraviolet regularization and keep the physical charge $e_{\rm ph}$ fixed, a phenomenon related to the Landau pole~\cite{Landau1}. The best we can do is to fix $e^2_{\rm ph}\log\Lambda$, which is equivalent to fixing $Z_3$, then let $e_{\rm ph}$ go to zero and verify that the chosen value of $Z_3$ does not appear in the perturbative series in $e_{\rm ph}$~\cite{GraLeSe2}. In this regime, our goal is to compare the total magnetic energy per unit volume
$$\frac{b^2}{8\pi}+f_{\rm vac}^{\rm PV}(e\,b)$$
with the physical energy involving the Euler-Heisenberg magnetic energy
$$\frac{(b_{\rm ph})^2}{8\pi}+f_{\rm vac}(e_{\rm ph}\,b_{\rm ph}).$$
The following says that these two energies are exponentially close to each other in a neighborhood of 0. In particular, they have the same Taylor expansion at $e_{\rm ph}=0$ which, however, is divergent as we have recalled above.

\begin{theorem}[Renormalization of $f_{\rm vac}^{\rm PV}$]
We have for a universal constant $K$
\begin{multline}
\left|\frac{b^2}{8\pi}+f_{\rm vac}^{\rm PV}(e\,b)-\frac{(b_{\rm ph})^2}{8\pi}-f_{\rm vac}(e_{\rm ph}\,b_{\rm ph})\right|\\ \leq K|c_1| \left(\frac{e_{\rm ph}b_{\rm ph}}{m}\right)^4 \exp\left(-6\pi\tfrac{1-Z_3}{(e_{\rm ph})^2}\right).
\label{eq:exponentially_close}
\end{multline}
\end{theorem}

\begin{proof}
Inserting~\eqref{eq:relation_PV_EH} and using the definition of $b_{\rm ph}$, we see that the quadratic terms cancel exactly and therefore obtain
\begin{multline}
 \frac{b^2}{8\pi}+f_{\rm vac}^{\rm PV}(e\,b)-\frac{(b_{\rm ph})^2}{8\pi}-f_{\rm vac}(e_{\rm ph}\,b_{\rm ph})\\
 =c_1\,\frac{m_1^4}{m^4}\,f_{\rm vac}\left(\frac{m^2}{m_1^2}eb\right)+c_2\,\frac{m_2^4}{m^4}\,f_{\rm vac}\left(\frac{m^2}{m_2^2}eb\right).
\end{multline}
We set 
$$K:=\max_{x\geq0}\frac{2m^4|f_{\rm vac}(x)|}{x^4}.$$
Although $f_{\rm vac}$ depends on $m$, it can be seen that $K$ is independent of $m$. We obtain
\begin{align*}
\left|\frac{b^2}{8\pi}+f_{\rm vac}^{\rm PV}(e\,b)-\frac{(b_{\rm ph})^2}{8\pi}-f_{\rm vac}(e_{\rm ph}\,b_{\rm ph})\right|&\leq \frac{K}2\left(-\frac{c_1}{(m_1)^4}+\frac{c_2}{(m_2)^4}\right)(e_{\rm ph}b_{\rm ph})^4\\
&\leq \frac{K|c_1|}{(m_1)^4}(e_{\rm ph}b_{\rm ph})^4,
\end{align*}
since $c_2=-1-c_1$ and $m_1<m_2$. Recall that
 \begin{align*}
-\frac{3\pi}{2}\frac{1-Z_3}{(e_{\rm ph})^2}=-\log\Lambda&= \log m +c_1\log m_1+c_2\log m_2\\
&\geq  \log m +(c_1+c_2)\log m_1=\log\frac{m}{m_1}
 \end{align*}
which gives us
$$\frac1{(m_1)^4}\leq \frac{e^{-6\pi\frac{1-Z_3}{(e_{\rm ph})^2}}}{m^4}$$
and concludes the proof.
\end{proof}

Our conclusion is that the Pauli-Villars-regulated function $f_{\rm vac}^{\rm PV}$ coincides with the original Euler-Heisenberg energy $f_{\rm vac}$ to any order, after charge renormalization. It would be interesting to relate the original Dirac energy $\boF_{\rm vac}$ with $f_{\rm vac}$ directly, by performing the charge renormalization at the same time as we take the limit $\varepsilon\to0$. We do not discuss this further, since this would force us to take $e\to0$.

\subsection{The semi-classical limit}\label{sec:semi-classics}
Let us now explain the main ideas behind the convergence in~\eqref{eq:conv-eps}. We stay here at a formal level. In the next section we give the detailed proof of Theorem~\ref{thm:lim-eps}, which takes a slightly different route. Using that
\begin{equation}
 |D_{m, 0,e A}|^2= \begin{pmatrix}
\mathscr{P}_{eA}+m^2 & 0\\
0 & \mathscr{P}_{eA}+m^2
\end{pmatrix}=\big(\mathscr{P}_{eA}+m^2\big)\otimes \1_{\C^2}
 \label{eq:relation_Pauli}
\end{equation}
with $V=0$ and where 
\begin{equation}
\mathscr{P}_{eA}=\Big(\sigma\cdot(-i\nabla+eA)\Big)^2=(-i\nabla+eA)^2-e\,\bssigma\cdot B
\label{eq:Pauli}
\end{equation}
is the Pauli operator, acting on 2-spinors, we can rewrite the vacuum energy as
\begin{equation*}
 \boF_{\rm vac}^{\rm PV}(0,e A_\varepsilon)=  \tr_{L^2(\R^3,\C^2)} \left\{\sum_{j = 0}^2 c_j \, \Big(\sqrt{\mathscr{P}_{0}+m_j^2} - \sqrt{\mathscr{P}_{eA_\varepsilon}+m_j^2}\Big)\right\}.
\end{equation*}
Next we use the integral formula
\begin{equation}
\sum_{j=0}^2c_j\sqrt{a+m_j^2}=-\frac{1}{2\sqrt\pi}\int_0^\infty \bigg(\sum_{j=0}^2c_je^{-sm_j^2}\bigg)e^{-sa}\,\frac{ds}{s^{3/2}}, 
 \label{eq:integral_repr_heat_kernels}
\end{equation}
for $a\geq0$, which is proved by an integration by parts. Note that the integral on the right converges at $s=0$, thanks to the Pauli-Villars condition~\eqref{eq:cond-PV}. We obtain
\begin{equation*}
 \boF_{\rm vac}^{\rm PV}(0,e A_\varepsilon)= \frac{1}{2\sqrt\pi} \tr_{L^2(\R^3,\C^2)} \int_0^\infty \bigg(\sum_{j=0}^2c_je^{-sm_j^2}\bigg)\Big(e^{-s\mathscr{P}_{eA_\varepsilon}}-e^{-s\mathscr{P}_{0}}\Big)\,\frac{ds}{s^{3/2}}.
\end{equation*}

Exchanging the trace with the integral, we need an expansion in $\varepsilon$ of $\tr(e^{-s\mathscr{P}_{eA_\varepsilon}}-e^{-s\mathscr{P}_{0}})$. Changing units, we recognize here a semi-classical limit. 
The idea is, therefore, to replace the trace by an integral over $\R^3$, of the value of this trace per unit volume for a constant field. 

We recall that for a constant magnetic field $B$, the operator $\mathscr{P}_{eA}$ commutes with the translations in the direction of $B$. For each value $\xi$ of the momentum in this direction, the corresponding fiber Hamiltonian has the energy levels
$$\big(2n+1+\nu)e|B|+m^2+\xi^2$$
where $n\geq0$ is the index of the Landau band and $\nu=\pm1$ is the spin variable. These energies are infinitely degenerate with respect to the angular momentum along the field $B$. The number of states with fixed energy per unit area in the plane perpendicular to $B$ is $e|B|/(2\pi)$. Semi-classical analysis then suggests that
\begin{multline}
\tr\big(e^{-s\mathscr{P}_{eA_\varepsilon}}-e^{-s\mathscr{P}_{0}}\big)\\
\simeq \int_{\R^3}\bigg\{\frac{e|B(\varepsilon x)|}{2\pi}\sum_{\nu=\pm1}\sum_{n\geq0}\int_\R e^{-s(2n+1+\nu)e|B(\varepsilon x)|}e^{-s\xi^2}\frac{d\xi}{2\pi}\\ 
- 2\int_{\R^3}e^{-s|p|^2}\,\frac{dp}{(2\pi)^3}\bigg\}dx.
\label{eq:to_be_shown_semi_classical}
\end{multline}
Using $\int_{\R^d}e^{-s|p|^2}\,dp=(\pi/s)^{d/2}$
and 
\begin{equation*}
 \sum_{\nu=\pm1}\sum_{n\geq0} e^{-s(2n+1+\nu)e|B|}=(1+e^{-2s|B|})\sum_{n\geq0} e^{-2sne|B|}=\coth(es|B|),
\end{equation*}
we can compute the right side of~\eqref{eq:to_be_shown_semi_classical} and obtain 
\begin{equation}
\tr\big(e^{-s\mathscr{P}_{eA_\varepsilon}}-e^{-s\mathscr{P}_{0}}\big)\simeq \frac{\varepsilon^{-3}}{4\pi^{3/2}s^{3/2}}\int_{\R^3}\big(es|B(x)|\coth(es|B(x)|)-1\big)dx.
\label{eq:semi-classics-heat-kernel}
\end{equation}
Inserting in~\eqref{eq:to_be_shown_semi_classical}, this is exactly the limit~\eqref{eq:conv-eps} stated in Theorem~\ref{thm:lim-eps}.

There are serious technical difficulties in the rigorous justification of these arguments. The first is that the trace is everywhere formal, since the operators are usually not trace-class. 
When $A$ is assumed to be integrable, the operators are trace-class and many of our arguments become simpler. But the integrability of $A$ will fail for physically relevant magnetic fields. The second difficulty is of course the justification of the semi-classical limit~\eqref{eq:semi-classics-heat-kernel}, which involves the whole spectrum of the two Pauli operators and not only the negative eigenvalues. 

In~\cite{Haba-82,Haba-84}, Haba has used the Feynman-Kac formula to study the limit of the kernel $(e^{-s\mathscr{P}_{eA}}-e^{-s\mathscr{P}_{0}})(x,x)$ for large fields but he did not state the semi-classical limit~\eqref{eq:semi-classics-heat-kernel}. He also used the strong assumption that $A\in L^1(\R^3)\cap L^\ii(\R^3)$ (see~\cite[Thm.~2]{Haba-82}). 

In the literature there are simple \emph{variational} arguments (e.g. based on coherent states~\cite{Lieb-73b,Lieb-81b,Simon-80}, including the case of magnetic fields~\cite{LieSoYn2, FuGuPeY1, Yngvaso1, ErdoSol1, ErdoSol2, ErdoSol3, Fournais-01, Fournais-01b, Fournais-02}) which provide the first order semi-classical term. So far these seem to have been mainly used for eigenvalues and it is not clear how to adapt them to our particular limit~\eqref{eq:semi-classics-heat-kernel} of a difference of two heat kernels. Instead, we will prove Theorem~\ref{thm:lim-eps} using another integral representation which involves the resolvent of the two Pauli operators instead of their heat kernels, as in~\cite{GrHaLeS1}. 

\section{Proof of Theorem~\ref{thm:lim-eps}}

In this section we describe the main steps of the proof of Theorem~\ref{thm:lim-eps}. Some technical details will be provided later. Everywhere we use the simpler notation $D_{m,A}$ instead of $D_{m,0,A}$, since there will be no electric potential $V$ from now on.

\subsubsection*{Step 1: Estimates on $A$}
We start with an elementary result which provides useful estimates on the magnetic potential $A$, for a magnetic field $B$ which satisfies the assumptions of Theorem~\ref{thm:lim-eps}.

\begin{lemma}[Estimates on $A$ and its derivatives]
\label{lem:hypo-B}
Let $B\in C^0(\R^3,\R^3)$ be such that $\div B=0$ and satisfying the assumptions~\eqref{eq:cond-B} of Theorem~\ref{thm:lim-eps}.
The magnetic potential 
\begin{equation}
\label{eq:Savart}
A(x) = - \frac{1}{4 \pi} \int_{\R^3} \frac{y}{|y|^3} \times B(x - y) \, dy=\frac{1}{4 \pi} \int_{\R^3} \frac{(\curl B)(x-y)}{|y|}  \, dy
\end{equation}
is the unique potential in $\dot{H}^1_{\rm div}(\R^3)$ such that $B=\curl A$. It is in $C^1(\R^3)$ and satisfies 
$$\begin{cases}
   A\in L^p(\R^3)&\text{for $3/2<p\leq\infty$,}\\
   \nabla A\in L^q(\R^3)&\text{for $1<q\leq\infty$}.
  \end{cases}$$
\end{lemma}

\begin{proof}
Since $y|y|^{-3}$ is in $L^{3/2}_{\rm w}(\R^3)$, the Hardy-Littlewood-Sobolev inequality gives that $A\in L^p(\R^3)$ for every $3/2<p<\ii$.
The case $p=\ii$ follows from the fact that $y|y|^{-3}\in L^1(\R^3)+L^2(\R^3)$ and $B\in L^2(\R^3)\cap L^\ii(\R^3)$.
Next we write
\begin{align*}
A(x)&=\frac{1}{4\pi}\int_{\R^3}\frac{1-e^{-|y|}}{|y|} \curl B(x-y)\,dy+\frac{1}{4\pi}\int_{\R^3}\frac{e^{-|y|}}{|y|} \curl B(x-y)\,dy\\
&=-\frac{1}{4\pi}\int_{\R^3}\nabla\bigg(\frac{1-e^{-|y|}}{|y|}\bigg) \times B(x-y)\,dy+\frac{1}{4\pi}\int_{\R^3}\frac{e^{-|y|}}{|y|} \curl B(x-y)\,dy.
\end{align*}
and obtain
\begin{multline*}
\partial_j A(x)= -\frac{1}{4\pi}\int_{\R^3}\partial_j\nabla\left(\frac{1-e^{-|y|}}{|y|}\right) \times B(x-y)\,dy\\
+\frac{1}{4\pi}\int_{\R^3}\partial_j\left(\frac{e^{-|y|}}{|y|}\right) \curl B(x-y)\,dy.
\end{multline*}
Since $\partial_j\left(|y|^{-1}e^{-|y|}\right)\in L^1(\R^3)\cap L^{5/4}(\R^3)$ and $\curl B\in L^1(\R^3)\cap L^5(\R^3)$, the second term is in $L^1(\R^3)\cap L^\ii(\R^3)$. Similarly, $\partial_j\nabla|y|^{-1}(1-e^{-|y|})\in L^p(\R^3)$ for all $p>1$ and $B\in L^1(\R^3)$ so the first term is in $L^p(\R^3)$ for all $p>1$.
\end{proof}

\subsubsection*{Step 2: Expression in terms of the resolvent of the Pauli operator}
Here we summarize some of the findings of~\cite{GrHaLeS1} which are useful for the proof of the main theorem.
Following~\cite{GrHaLeS1}, we use the integral representation 
$$|x| = \frac{1}{\pi} \int_\R \frac{x^2}{x^2 + \omega^2} \, d\omega=\frac{1}{\pi} \int_\R \left(1-\frac{\omega^2}{x^2 + \omega^2}\right) \, d\omega$$
which allows to express the operator $T_{eA}$ as
\begin{align}
T_{eA}&=\sum_{j=0}^2c_j  \, \Big( \big| D_{m_j, 0} \big| - \big| D_{m_j, e  A} \big| \Big)\nonumber\\
&=\frac{1}{\pi} \int_\R \sum_{j=0}^2c_j  \left(\frac{1}{\big| D_{m_j, e  A} \big|^2 + \omega^2}-\frac{1}{\big| D_{m_j, 0} \big|^2 + \omega^2}\right) \omega^2d\omega.\label{eq:abs-val}
\end{align}
Using the two Pauli-Villars conditions in~\eqref{eq:cond-PV}, a calculation gives the following useful formulas:
\begin{align}
\forall x\geq0,\qquad \sum_{j=0}^2c_j\frac{1}{x+m_j^2}&=\sum_{j=0}^2c_j\left(\frac{1}{x+m_j^2}-\frac{1}{x+m^2}\right)\nonumber\\
&=\sum_{j=0}^2c_j(m_j^2-m^2)^2\frac{1}{(x+m_j^2)(x+m^2)^2}.
\label{eq:annulations_sum}
\end{align}
These formulas show how the Pauli-Villars conditions~\eqref{eq:cond-PV} allow to increase the decay at large momenta, since $\sum_{j=0}^2c_j(x+m_j^2)^{-1}$ now behaves like $x^{-3}$ at infinity instead of $x^{-1}$ for the unregularized resolvent $(x+m^2)^{-1}$. 
Using~\eqref{eq:annulations_sum} and the bound $|D_{m_j,eA}|^2\geq m_j^2$ following from~\eqref{eq:relation_Pauli}, we deduce that 
$$\norm{\sum_{j=0}^2c_j  \left(\frac{1}{\big| D_{m_j, e  A} \big|^2 + \omega^2}-\frac{1}{\big| D_{m_j, 0} \big|^2 + \omega^2}\right)}\leq \frac{C}{(m^2+\omega^2)^3}.$$
Therefore the integral in~\eqref{eq:abs-val} converges in the operator norm. Using~\eqref{eq:relation_Pauli}, we find that $T_{eA}$ is block-diagonal in the two-spinor basis, which we can rewrite as
$$T_{eA}=\frac{1}{\pi} \int_\R\omega^2d\omega \sum_{j=0}^2c_j  \left(\frac{1}{\mathscr{P}_{e  A} +m_j^2 + \omega^2}-\frac{1}{\mathscr{P}_{0} +m_j^2 + \omega^2}\right)\otimes \1_{\C^2}$$
where $\mathscr{P}_{eA}$ is the Pauli operator defined in~\eqref{eq:Pauli}. The integral representation used in the previous section has the nice feature that the regularization gives rise to the simple term $\sum_{j=0}^2 c_j e^{-sm_j^2}$. Here the regularization parameters cannot be separated from the other terms in the integral. Nevertheless, resolvents are simpler to handle and we will be able to apply several estimates proved in~\cite{GrHaLeS1} in this setting.

To be more precise, it was proved in~\cite{GrHaLeS1} that when $A\in L^1(\R^3)\cap H^1_{\rm div}(\R^3)$ and $B\in L^1(\R^3)\cap L^2(\R^3)$, the operator 
$$\sum_{j=0}^2c_j\left(\frac{1}{\mathscr{P}_{e  A} +m_j^2  + \omega^2}-\frac{1}{\mathscr{P}_{0} +m_j^2  + \omega^2}\right)$$
is trace-class and that the trace can be integrated over $\omega$:
$$\int_{\R}\tr\left|\sum_{j=0}^2c_j\left(\frac{1}{\mathscr{P}_{e  A} +m_j^2  + \omega^2}-\frac{1}{\mathscr{P}_{0} +m_j^2  + \omega^2}\right)\right|\omega^2d\omega<\infty$$
which then proved that $T_{eA}$ is itself trace-class, as stated in Theorem~\ref{thm:def-L-PV}. In our situation, we have $B\in L^1\cap L^\ii$ but $A$ is not necessarily integrable, so the operator is not necessarily trace-class. However, the non-trace-class part is easy to extract and happens to be linear in $A$. This is the content of the following

\begin{prop}[Extracting the non-trace-class part]\label{prop:extract_non_trace_class}
We assume that $B$ satisfies the conditions~\eqref{eq:cond-B} of Theorem~\ref{thm:lim-eps} and let $A$ be the unique potential in $H^1_{\rm div}(\R^3)$ such that $B=\curl A$. Then
\begin{multline}
\int_\R\tr\bigg|\sum_{j=0}^2c_j \bigg(\frac{1}{\mathscr{P}_{eA} +m_j^2 + \omega^2}-\frac{1}{\mathscr{P}_{0} +m_j^2  + \omega^2}\\-\frac{1}{\mathscr{P}_{0} +m_j^2  + \omega^2}(p\cdot A+A\cdot p)\frac{1}{\mathscr{P}_{0} +m_j^2 + \omega^2}\bigg)\bigg|\, \omega^2d\omega<\ii
\label{eq:estim_integral_absolute}
\end{multline}
and
\begin{multline}
\boF_{\rm vac}^{\rm PV}(eA)=\frac{1}{\pi}\int_\R\tr\bigg\{\sum_{j=0}^2c_j \bigg(\frac{1}{\mathscr{P}_{eA} +m_j^2 + \omega^2}-\frac{1}{\mathscr{P}_{0} +m_j^2  + \omega^2}\\-\frac{1}{\mathscr{P}_{0} +m_j^2  + \omega^2}(p\cdot A+A\cdot p)\frac{1}{\mathscr{P}_{0} +m_j^2 + \omega^2}\bigg)\bigg\}\, \omega^2d\omega,
\label{eq:integral_formula_resolvents}
\end{multline}
where $p:=-i\nabla$. If additionally $A\in L^1(\R^3)$, then 
$$\tr\bigg\{\sum_{j=0}^2c_j\frac{1}{\mathscr{P}_{0} +m_j^2  + \omega^2}(p\cdot A+A\cdot p)\frac{1}{\mathscr{P}_{0} +m_j^2 + \omega^2}\bigg\}=0$$
for every $\omega\in\R$, where the operator in the trace is trace-class.
\end{prop}

Note that $p\cdot A+A\cdot p=2p\cdot A=2A\cdot p$ since $\div A=0$.

\begin{proof}
We only give a sketch of the proof, which relies on the techniques used in~\cite{GrHaLeS1}. It is based on the 5th-order resolvent expansion
\begin{align}
&\sum_{j=0}^2c_j\left(\frac{1}{\mathscr{P}_{e  A} +m_j^2  + \omega^2}-\frac{1}{\mathscr{P}_{0} +m_j^2  + \omega^2}\right)\nonumber\\
&\qquad =\sum_{n=1}^4(-1)^n\sum_{j=0}^2c_j\left(\frac{1}{\mathscr{P}_{0} +m_j^2  + \omega^2}S_{eA}\right)^n\frac{1}{\mathscr{P}_{0} +m_j^2  + \omega^2}\nonumber\\
&\qquad\quad -\sum_{j=0}^2c_j\left(\frac{1}{\mathscr{P}_{0} +m_j^2  + \omega^2}S_{eA}\right)^5\frac{1}{\mathscr{P}_{eA} +m_j^2  + \omega^2}\nonumber\\
&\qquad:=\sum_{n=1}^4T_{eA}^{(n)}(\omega)+T_{eA}^{(5)}(\omega).\label{eq:resolvent_expansion_free}
\end{align}
where the operator $S_{e A}$ is defined by
\begin{align}
S_{e A} := \mathscr{P}_{e  A}-\mathscr{P}_{0}&=-e (p\cdot A+A\cdot p) + e^2 |A|^2 - e B \cdot \bssigma\label{def:SeA}\\
&=-e (\bssigma\cdot p) (\bssigma\cdot A)-e (\bssigma\cdot A)(\bssigma\cdot p) + e^2 |A|^2. 
\label{def:SeA_bis}
\end{align}
We emphasize that $T_{eA}^{(5)}(\omega)$ contains $\mathscr{P}_{eA}$ in the last resolvent on the right, whereas $T_{eA}^{(n)}(\omega)$ has $\mathscr{P}_{0}$ for $n\leq4$. This should note generate any confusion, since we will never introduce $T_{eA}^{(n)}(\omega)$ for $n>5$. 
For $n\leq5$, it was proved in~\cite{GrHaLeS1} that $T_{eA}^{(n)}(\omega)$  is trace-class, with 
$$\int \tr\big|T_{eA}^{(n)}(\omega)\big|\;\omega^2 d\omega<\infty$$
whenever $A\in L^n(\R^3)\cap \dot{H}^1_{\rm div}(\R^3)$ and $B\in L^2(\R^3)$. From Lemma~\ref{lem:hypo-B}, we have $A\in L^2(\R^3)\cap L^\infty(\R^3)$ and $B\in L^1(\R^3)\cap L^\infty(\R^3)$ hence only the first order term $T_{eA}^{(1)}(\omega)$ is possibly not trace-class under the condition~\eqref{eq:cond-B}. All the other terms are trace-class. The term involving $p\cdot A+A\cdot p$ which we have subtracted in the statement is exactly the non-trace-class part of $T^{(1)}_{eA}$, as we will explain below. Formula~\eqref{eq:integral_formula_resolvents} follows from the continuity of $\boF^{\rm PV}_{\rm vac}$ in Theorem~\ref{thm:def-L-PV}.

Let us briefly explain how to estimate $T^{(n)}_{eA}(\omega)$ using ideas from~\cite{GrHaLeS1}. For the 5th order term, we use directly the Kato-Seiler-Simon inequality (see~\cite[Thm. 4.1]{Simon01}) 
\begin{equation}
\label{eq:KSS}
\big\| g(-i \nabla)f(x)  \big\|_{\gS_p}=\big\| f(x) g(-i \nabla) \big\|_{\gS_p} \leq (2 \pi)^{- \frac{3}{p}} \big\| f \big\|_{L^p} \big\| g \big\|_{L^p},
\end{equation}
which holds for any number $p \geq 2$, and any functions $(f , g) \in L^p(\R^3)^2$. Here, the notation $\| \cdot \|_{\gS_p}$ stands for the norm of the Schatten class $\gS_p(L^2(\R^3, \C^4))$. Using H\"older's inequality in Schatten spaces and the positivity of $\mathscr{P}_{eA}\geq0$, we find
\begin{equation*}
\norm{T^{(5)}_{eA}(\omega)}_{\gS_1}\leq \sum_{j=0}^2\frac{|c_j|}{m_j^2+\omega^2} \norm{\frac{1}{\mathscr{P}_{0} +m_j^2  + \omega^2}S_{eA}}_{\gS_5}^5.
\end{equation*}
Here it is convenient to use that 
\begin{equation}
S_{e A}=-2e p\cdot A + e^2 |A|^2 - e B \cdot \bssigma
\label{eq:2nd_def_S_eA}
\end{equation}
since $\div A=0$. 
By the triangle inequality and the fact that $\norm{\cdot}_{\gS_p}\leq\norm{\cdot}_{\gS_q}$ for $q\leq p$, we obtain
\begin{align*}
&\norm{\frac{1}{\mathscr{P}_{0} +m_j^2  + \omega^2}S_{eA}}_{\gS_5}\\
&\qquad\leq 2e\norm{\frac{p}{p^2 +m_j^2  + \omega^2}\cdot A(x)}_{\gS_5}+e^2\norm{\frac{1}{p^2 +m_j^2  + \omega^2}|A(x)|^2}_{\gS_3}\\
&\qquad\qquad+e\norm{\frac{1}{p^2 +m_j^2  + \omega^2}|B(x)|}_{\gS_2}\\
&\qquad\lesssim \frac{e\norm{A}_{L^5}}{(m_j^2+\omega^2)^{1/5}}+\frac{e^2\norm{A}^2_{L^6}}{(m_j^2+\omega^2)^{1/2}}+\frac{e\norm{B}_{L^2}}{(m_j^2+\omega^2)^{1/4}}
\end{align*}
which, by the Sobolev inequality, gives the final estimate
$$\norm{T^{(5)}_{eA}(\omega)}_{\gS_1}\lesssim \frac{\sum_{j=0}^2|c_j|}{m^2+\omega^2}\left(\frac{e^5\norm{A}_{H^1}^5}{m^2+\omega^2}+\frac{e^{10}\norm{A}_{H^1}^{10}}{(m^2+\omega^2)^{5/2}}\right)$$
and proves that $\int_\R \|T^{(5)}_{eA}(\omega)\|_{\gS_1}\omega^2\,d\omega<\ii$, as we wanted.

The argument for the lower order terms $T^{(n)}_{eA}(\omega)$ is slightly more complicated, since we need to use the Pauli-Villars condition~\eqref{eq:cond-PV} in order to increase the decay in momentum. Let us start with the proof for the first order term $T^{(1)}_{eA}(\omega)$. The idea is to insert the resolvent with the mass $m=m_0$ using the Pauli-Villars conditions~\eqref{eq:cond-PV}. For shortness, it is convenient to introduce the notation
\begin{equation}
K_j(\omega):=\mathscr{P}_{0}+m_j^2  + \omega^2=-\Delta+m_j^2  + \omega^2,
\label{def:K_j}
\end{equation}
for the Klein-Gordon operator.
We then use the relation
\begin{align}
-T_{eA}^{(1)}(\omega)&=\sum_{j=0}^2c_j\frac{1}{K_j(\omega)}S_{eA}\frac{1}{K_j(\omega)}\nonumber\\
&=\sum_{j=1}^2c_j(m_j^2-m^2)^2\frac{1}{K_j(\omega)K_0(\omega)}S_{eA}\frac{1}{K_j(\omega)K_0(\omega)}\nonumber\\
&\qquad+\sum_{j=1}^2c_j(m_j^2-m^2)^2\frac{1}{K_0(\omega)}S_{eA}\frac{1}{K_j(\omega)K_0(\omega)^2}\nonumber\\
&\qquad+\sum_{j=1}^2c_j(m_j^2-m^2)^2\frac{1}{K_j(\omega)K_0(\omega)^2}S_{eA}\frac{1}{K_0(\omega)}.\label{eq:annulations}
\end{align}
Arguing as before using the Kato-Seiler-Simon inequality~\eqref{eq:KSS}, we obtain that the second and third terms in the definition~\eqref{def:SeA} of $S_{eA}$ are trace-class with
\begin{multline*}
\norm{\sum_{j=0}^2c_j\frac{1}{K_j(\omega)}\big(e^2 |A|^2 - e B \cdot \bssigma\big)\frac{1}{K_j(\omega)}}_{\gS_1}\\
\lesssim \left(|c_1|(m_1^2-m^2)^2+|c_2|(m_2^2-m^2)^2\right)\frac{\norm{A}_{L^2}^2+\norm{B}_{L^1}}{(m^2+\omega^2)^{5/2}}.
\end{multline*}
The term involving $p\cdot A+A\cdot p$ can be treated in the same way, but it is only trace-class when $A\in L^1(\R^3)$. It is not trace-class under our assumptions on $A$ and this is the term which has been subtracted in the statement of the proposition. Fortunately, when $A\in L^1(\R^3)$ its trace vanishes due to the invariance under complex conjugation. Namely, the operator
$$\sum_{j=0}^2c_j\frac{1}{K_j(\omega)}(p\cdot A+A\cdot p)\frac{1}{K_j(\omega)}$$
is self-adjoint, so its trace is real. Applying complex conjugation we find 
$$\overline{\sum_{j=0}^2c_j\frac{1}{K_j(\omega)}(p\cdot A+A\cdot p)\frac{1}{K_j(\omega)}}=-\sum_{j=0}^2c_j\frac{1}{K_j(\omega)}(p\cdot A+A\cdot p)\frac{1}{K_j(\omega)}$$
since $\overline{p}=-p$ and $K_j(\omega)$ and $A$ are both real. So its trace is imaginary, and thus equal to 0.

The proof that the second order term $T^{(2)}_{eA}(\omega)$ is trace-class under our assumptions~\eqref{eq:cond-B} on $B$ is similar and relies on the following identity:
\begin{align}
 T^{(2)}_{eA}(\omega) = &\sum_{j= 0}^2 c_j \, (m_j^2 - m_0^2)^2  \Big( \frac{1}{K_j(\omega)K_0(\omega)^2} \, S_{e A} \, \frac{1}{K_j(\omega)} \, S_{e A} \, \frac{1}{K_j(\omega)}\nonumber\\
& \qquad + \frac{1}{K_0(\omega)^2} \, S_{e A} \, \frac{1}{K_j(\omega)K_0(\omega)} \, S_{e A} \, \frac{1}{K_j(\omega)}\nonumber\\
& \qquad + \frac{1}{K_0(\omega)^2} \, S_{e A} \, \frac{1}{K_0(\omega)} \, S_{e A} \, \frac{1}{K_j(\omega)K_0(\omega)}\nonumber\\
& \qquad + \frac{1}{K_0(\omega)} \, S_{e A} \, \frac{1}{K_j(\omega)K_0(\omega)^2} \, S_{e A} \, \frac{1}{K_j(\omega)}\nonumber\\
& \qquad + \frac{1}{K_0(\omega)} \, S_{e A} \, \frac{1}{K_0(\omega)^2} \, S_{e A} \, \frac{1}{K_j(\omega)K_0(\omega)}\nonumber\\
& \qquad + \frac{1}{K_0(\omega)} \, S_{e A} \, \frac{1}{K_0(\omega)} \, S_{e A} \, \frac{1}{K_j(\omega)K_0(\omega)^2} \Big).
\label{eq:dev-gT2}
\end{align}

The third and fourth order terms $T^{(3)}_{eA}(\omega)$ and $T^{(4)}_{eA}(\omega)$ are somewhat easier to handle since only the first Pauli-Villars condition in~\eqref{eq:cond-PV} is necessary. For the third-order term we use
\begin{equation}
\label{eq:dev-gT3}
\begin{split}
T^{(3)}_{eA}(\omega) & = \sum_{j= 0}^2 c_j \, (m_0^2 - m_j^2)  \Big( \frac{1}{K_j(\omega)} \, \frac{1}{K_0(\omega)} \, \Big( S_{e A} \, \frac{1}{K_j(\omega)} \Big)^3 \\
&+ \frac{1}{K_0(\omega)} \, S_{e A} \frac{1}{K_j(\omega)} \, \frac{1}{K_0(\omega)} \, \Big( S_{e A} \, \frac{1}{K_j(\omega)} \Big)^2 \\
&+ \Big( \frac{1}{K_0(\omega)} \, S_{e A} \Big)^2 \, \frac{1}{K_j(\omega)} \, \frac{1}{K_0(\omega)} \, S_{e A} \, \frac{1}{K_j(\omega)}\\
& + \Big( \frac{1}{K_0(\omega)} \, S_{e A} \Big)^3 \, \frac{1}{K_j(\omega)} \, \frac{1}{K_0(\omega)} \Big) .
\end{split}
\end{equation}
The fourth order term is similar and this concludes the proof of Proposition~\ref{prop:extract_non_trace_class}.
\end{proof}

\subsubsection*{Step 3: Localization}
Since we will be considering slowly varying potentials, it will be useful to localize our energy to sets of fixed size $\rho$, where the magnetic field will be essentially constant. We introduce the Gaussian function $G_\rho$ given by
$$G_\rho(x) := (\pi \rho)^{- \frac{3}{2}} e^{- \frac{|x|^2}{\rho^2}},$$
and recall that
\begin{equation}
\label{eq:Grho}
\int_{\R^3} G_\rho(x - y)^2 \, dy = 1,
\end{equation}
for any $x \in \R^3$, which we interpret as a continuous partition of unity. The following is well-known.

\begin{lemma}[Localization of a trace-class operator]\label{lem:localization_abstract}
Let $T$ be a trace-class self-adjoint operator on $L^2(\R^3,\C^2)$ and $G_\rho(\cdot-y)$ be the multiplication operator by the function $x\mapsto G_\rho(x-y)$. Then $G_\rho(\cdot-y) T G_\rho(\cdot-y)$ is also trace-class with
$$\int_{\R^3}\tr|G_\rho(\cdot-y) T G_\rho(\cdot-y)|\,dy\leq \tr|T|<\ii$$
and
$$\tr T=\int_{\R^3}\tr G_\rho(\cdot-y) T G_\rho(\cdot-y)\,dy=\int_{\R^3}\tr \big(G_\rho \tau_{-y}T\tau_{y} G_\rho\big)\,dy$$
where $(\tau_y f)(x)=f(x-y)$ is the unitary operator which translates by $y\in\R^3$.
\end{lemma}

\begin{proof}
It is clear that $G_\rho(\cdot-y) T G_\rho(\cdot-y)$ is trace-class since $x\mapsto G_\rho(x-y)$ is bounded, hence defines a bounded operator. Now, we can diagonalize $T=\sum_{j\geq1} t_j|u_j\rangle\langle u_j|$ with $\sum_{j\geq1}|t_j|<\ii$ and obtain
$$G_\rho(\cdot-y) T G_\rho(\cdot-y)=\sum_{j\geq1} t_j|G_\rho(\cdot -y)u_j\rangle\langle u_jG_\rho(\cdot -y)|$$
with the sum being convergent in the trace-class.
In particular, by the triangle inequality
\begin{align*}
\tr|G_\rho(\cdot-y) T G_\rho(\cdot-y)|&\leq \sum_{j\geq1} |t_j| \tr |G_\rho(\cdot -y)u_j\rangle\langle u_jG_\rho(\cdot -y)|\\
&=\left(G_\rho\ast \sum_{j\geq1} |t_j|\; |u_j|^2\right)(y). 
\end{align*}
The rest follows from the fact that $\int_{\R^3}\sum_{j\geq1} |t_j|\; |u_j|^2=\sum_{j\geq1}|t_j|=\tr|T|$.
\end{proof}

Using this lemma, we can localize the operator appearing in the parenthesis in~\eqref{eq:integral_formula_resolvents} and obtain, after changing $y$ into $\varepsilon y$ and using the translation-invariance of $\mathscr{P}_0=-\Delta$,
\begin{equation}
\boxed{\boF_{\rm vac}^{\rm PV}(eA_\varepsilon)=\frac{\varepsilon^{-3}}{\pi}\int_\R\omega^2d\omega\int_{\R^3}dy\; f_{\omega}(eA_{\varepsilon,y})}
\label{eq:integral_formula_resolvents_y}
\end{equation}
with 
$$A_{\varepsilon,y}(x)=\varepsilon^{-1}A(y+\varepsilon x)$$ 
and
\begin{multline}
f_{\omega}(A)= \tr_{L^2(\R^3,\C^2)}\bigg\{ G_\rho\sum_{j=0}^2c_j \bigg(\frac{1}{\mathscr{P}_{A} +m_j^2 + \omega^2}-\frac{1}{\mathscr{P}_{0} +m_j^2  + \omega^2}\\-\frac{1}{\mathscr{P}_{0} +m_j^2  + \omega^2}(p\cdot A+A\cdot p)\frac{1}{\mathscr{P}_{0} +m_j^2 + \omega^2}\bigg)G_\rho\bigg\}.
\end{multline}
From~\eqref{eq:estim_integral_absolute} and  Lemma~\ref{lem:localization_abstract} we also know that
$$\int_\R\omega^2d\omega\int_{\R^3}dy\;\big|f_{\omega}(eA_{\varepsilon,y})\big|<\ii.$$
Since the localization length $\rho$ will not play an important role, we do not mention it in our notations.

The reason for subtracting the term with $p\cdot A+A\cdot p$ was that it is not trace-class, although its trace vanishes when $A\in L^1(\R^3)$. With the localization $G_\rho$, this operator has now become trace-class and it still has a vanishing trace, which simplifies a bit our reasoning.

\begin{lemma}
For every fixed $\rho>0$, and every $A\in {H}_{\rm div}^1(\R^3)$, the operator
$$K=\sum_{j=0}^2c_j\,G_\rho\frac{1}{\mathscr{P}_{0} +m_j^2  + \omega^2}(p\cdot A+A\cdot p)\frac{1}{\mathscr{P}_{0} +m_j^2 + \omega^2}G_\rho$$
is trace-class and its trace vanishes.
\end{lemma}

\begin{proof}
Using~\eqref{eq:annulations} and the Kato-Seiler-Simon inequality~\eqref{eq:KSS}, it is easy to see that $K$
is trace-class. The trace vanishes for the same reason as without $G_\rho$, namely $K^*=K$ and $\overline{K}=-K$.
\end{proof}

As a corollary, we immediately deduce that
\begin{equation}
f_{\omega}(eA_{\epsilon,y})= \tr_{L^2(\R^3,\C^2)}\bigg\{ G_\rho\sum_{j=0}^2c_j \bigg(\frac{1}{\mathscr{P}_{eA_{\epsilon,y}} +m_j^2 + \omega^2}-\frac{1}{\mathscr{P}_{0} +m_j^2  + \omega^2}\bigg)G_\rho\bigg\}
\label{eq:def_f_y_epsilon_omega}
\end{equation}
with the operator being trace-class under our assumptions~\eqref{eq:cond-B} on $B$.

\subsubsection*{Step 4: Localized energy with unbounded potentials}

Using the localization $G_\rho$, the idea is now to replace in~\eqref{eq:def_f_y_epsilon_omega} the potential $A_{\varepsilon,y}(x)=A(y+\varepsilon x)/\epsilon$ by the potential $B(y)\times x/2$ of the constant field $B(y)$, up to a small error. On the contrary to all the potentials we have considered so far, the latter has a linear growth at infinity. It is therefore necessary to extend the definition of $f_\omega$ to potentials growing at infinity, which is the content of the following result.

\begin{prop}[Uniformly bounded magnetic fields]\label{prop:growth}
\label{lem:trace-class}
Let $\rho > 0$. Consider a magnetic potential $A \in C^1(\R^3, \R^3)$ for which $B=\curl A\in L^\ii(\R^3)$ and let $\mathscr{P}_A$ be the Friedrichs extension of the corresponding Pauli operator. Then, for every $\omega\in\R$ the operator 
$$G_\rho\sum_{j=0}^2c_j \frac{1}{\mathscr{P}_{A} +m_j^2 + \omega^2}G_\rho$$
is trace-class. In particular, we can define 
\begin{multline}
f_{\omega}(A)=\tr_{L^2(\R^3,\C^2)}\left\{G_\rho\sum_{j=0}^2c_j \bigg(\frac{1}{\mathscr{P}_{A} +m_j^2 + \omega^2}-\frac{1}{\mathscr{P}_{0} +m_j^2  + \omega^2}\bigg)G_\rho\right\}.
\label{eq:def_f_omega_A}
\end{multline}
This is a gauge-invariant functional: for any function $\theta\in C^2(\R^3,\R)$, we have
\begin{equation}
f_\omega(A)=f_\omega(A+\nabla\theta).
\label{eq:gauge-invariance}
\end{equation}
\end{prop}

\begin{remark}\it
Using the exponential decay of $G_\rho$, it is possible to generalize the result to any potential $A$ for which $B$ has a polynomial growth at infinity, but we do not discuss this further.
\end{remark}

\begin{proof}[Proof of Proposition~\ref{prop:growth}]
The proof of Proposition~\ref{lem:trace-class} is based on the following observation.

\begin{lemma}[Positivity]\label{lem:positivity}
Assume that the coefficients $c_0 = 1$, $c_1$ and $c_2$, and the masses $0<m = m_0 < m_1 < m_2$ satisfy the Pauli-Villars conditions~\eqref{eq:cond-PV}. Then
$$\sum_{j=0}^2c_j f(m_j^2)\geq0$$
for every convex function on $[0,\ii)$. In particular
$$\sum_{j=0}^2\frac{c_j}{s+m_j^2}\geq0\quad\text{and}\quad \sum_{j=0}^2 c_j e^{-sm_j^2}\geq0,\quad\text{for every $s\geq0$.}$$
\end{lemma}

\begin{proof}[Proof of Lemma~\ref{lem:positivity}]
By convexity we have, using the Pauli-Villars condition~\eqref{eq:cond-PV},
$$f(m^2)+c_2f(m_2^2)\geq (1+c_2)f\left(\frac{m^2+c_2m_2^2}{1+c_2}\right)= -c_1f(m_1^2).$$
\end{proof}

By Lemma~\ref{lem:positivity} we have $\sum_{j=0}^2c_j (\mathscr{P}_{A} +m_j^2 + \omega^2)^{-1}\geq0$, in the sense of operators, and therefore the same is true with $G_\rho$ on both sides. The trace is then always well-defined in $[0,\ii]$ and it suffices to estimate it in order to prove that the operator is trace-class. We use~\eqref{eq:annulations_sum}, $c_1\leq0$, and obtain 
\begin{align*}
 \sum_{j=0}^2c_j \frac{1}{x+m_j^2 + \omega^2}
 &= \sum_{j=0}^2c_j(m_j^2-m^2)^2\frac{1}{x +m_j^2  + \omega^2}\frac{1}{(x +m^2  + \omega^2)^2}\\
  &\leq  \frac{c_2(m_2^2-m^2)^2}{m_2^2+\omega^2}\frac{1}{(x +m^2  + \omega^2)^2},
\end{align*}
hence
$$\sum_{j=0}^2c_j \frac{1}{\mathscr{P}_A+m_j^2 + \omega^2}\leq \frac{c_2(m_2^2-m^2)^2}{m_2^2+\omega^2}\frac{1}{(\mathscr{P}_A +m^2  + \omega^2)^2}$$
in the sense of operators. We can multiply by $G_\rho$ on both sides without changing the inequality and obtain after taking the trace
$$\tr \left\{G_\rho\sum_{j=0}^2c_j \frac{1}{\mathscr{P}_{A}+m_j^2 + \omega^2}G_\rho\right\} \leq \frac{c_2(m_2^2-m^2)^2}{m_2^2+\omega^2}\norm{G_\rho\frac{1}{\mathscr{P}_{A} +m^2  + \omega^2}}_{\gS_2}^2.$$
It is therefore sufficient to show that $G_\rho(\mathscr{P}_A+m^2+\omega^2)^{-1}$ is a Hilbert-Schmidt operator, under the assumptions of the proposition. This is the content of the following lemma.

\begin{lemma}[Magnetic Kato-Seiler-Simon-type inequality]\label{lem:KSS_magnetic}
Consider a magnetic potential $A \in C^1(\R^3, \R^3)$ for which $B=\curl A\in L^\ii(\R^3)$.
Then we have, for every $2\leq p\leq\ii$ and every $\mu>0$
\begin{equation}
\norm{f(x)(\mathscr{P}_A+\mu)^{-1}}_{\gS_p}\leq 2^{-\frac2p}\pi^{-\frac1p}\left(1+\frac{\|B\|_{L^\ii}}{\mu}\right)^{\frac{2}{p}}\mu^{\frac{3}{2p}-1}\norm{f}_{L^p}.\label{eq:estim_KSS_magnetic}
\end{equation}
\end{lemma}

\begin{proof}
Since $\mathscr{P}_A\geq0$, we have $\|(\mathscr{P}_A+\mu)^{-1}\|\leq \mu^{-1}$ and the bound is obvious for $p=\ii$. For $p=2$ we write
$$\frac{1}{\mathscr{P}_A+\mu}=\frac{1}{\mathscr{P}_A+\mu+\|B\|_{L^\ii}}\times\frac{\mathscr{P}_A+\mu+\|B\|_{L^\ii}}{\mathscr{P}_A+\mu}$$
and use that
$$\norm{\frac{\mathscr{P}_A+\mu+\|B\|_{L^\ii}}{\mathscr{P}_A+\mu}}\leq 1+\frac{\|B\|_{L^\ii}}{\mu}.$$
Therefore, it suffices to estimate the Hilbert-Schmidt norm of $f(\mathscr{P}_A+\mu+\|B\|_{L^\ii})^{-1}$. Now we use the fact that the integral kernel of $(\mathscr{P}_A+\mu+\|B\|_{L^\ii})^{-1}$ is pointwise bounded by that of $(-\Delta+\mu)^{-1}$:
\begin{equation}
\big| (\mathscr{P}_A+\mu+\|B\|_{L^\ii})^{-1}(x,y)\big|_\ii\leq \big| (-\Delta+\mu)^{-1}(x,y)\big|_\ii
\label{eq:pointwise_bd_kernel_resolvents}
\end{equation}
where $|\cdot|_\ii$ is the sup norm of $2\times2$ hermitian matrices.
The proof of~\eqref{eq:pointwise_bd_kernel_resolvents} is well known and goes as follows. First we reduce it to the similar pointwise bound on the heat kernels
\begin{equation}
\big|e^{-s(\mathscr{P}_A+\mu+\|B\|_{L^\ii})}(x,y)\big|_\ii\leq e^{-s(-\Delta+\mu)}(x,y)
\label{eq:pointwise_bd_kernel_heat}
\end{equation}
using the integral formula $H^{-1}=\int_0^\ii e^{-sH}\,ds$. Then, we use Trotter's formula
\begin{align}
& e^{-s(\mathscr{P}_A+\mu+\|B\|_{L^\ii})}(x,y)\nonumber\\
&\qquad=\lim_{n\to\ii}\left(e^{-\frac{s}n(-i\nabla+A)^2+\mu)}e^{-\frac{s}{n}(-\bssigma\cdot B+\|B\|_{L^\ii})}\right)^n(x,y)\nonumber\\
&\qquad=\lim_{n\to\ii}\int_{\R^3}dx_1\cdots\int_{\R^3}dx_{n-1}\; e^{-\frac{s}n((-i\nabla+A)^2+\mu)}(x,x_1)\times \nonumber\\
&\qquad\qquad \times e^{-\frac{s}{n}(-\bssigma\cdot B(x_1)+\|B\|_{L^\ii})}e^{-\frac{s}n((-i\nabla+A)^2+\mu)}(x_1,x_2)\times\cdots\nonumber\\ 
 &\qquad\qquad\cdots \times e^{-\frac{s}n((-i\nabla+A)^2+\mu)}(x_{n-1},y)e^{-\frac{s}{n}(-\bssigma\cdot B(y)+\|B\|_{L^\ii})}.
 \label{eq:Trotter}
\end{align}
The estimate~\eqref{eq:pointwise_bd_kernel_heat} follows from the diamagnetic inequality~\cite{Simon-05}
$$|e^{-\frac{s}n((-i\nabla+A)^2+\mu)}(x,y)|_\ii\leq e^{-\frac{s}n(-\Delta+\mu)}(x,y)$$
and the fact that
$$|e^{-\frac{s}{n}(-\bssigma\cdot B(x)+\|B\|_{L^\ii})}|_\ii\leq 1,$$
since $-\bssigma\cdot B(x)+\|B\|_{L^\ii}$ is a non-negative $2\times2$ symmetric matrix for every $x\in\R^3$. We therefore conclude that
\begin{align*}
 &\norm{f(x)(\mathscr{P}_A+\mu+\|B\|_{L^\ii})^{-1}}_{\gS_2}^2\\
 &\qquad =\int_{\R^3}\int_{\R^3}|f(x)|^2\; \big| (\mathscr{P}_A+\mu+\|B\|_{L^\ii})^{-1}(x,y)\big|_2^2\,dx\,dy \\
 &\qquad \leq 2\norm{f(x)(-\Delta+\mu)^{-1}}_{\gS_2}^2=2(2\pi)^{-3}\norm{f}_{L^2}^2\int_{\R^3}\frac{dp}{(p^2+\mu)^2}=\frac{\norm{f}_{L^2}^2}{4\pi \sqrt{\mu}}
\end{align*}
and
\begin{equation}
\norm{f(x)(\mathscr{P}_A+\mu)^{-1}}_{\gS_2}\leq \left(1+\frac{\|B\|_{L^\ii}}{\mu}\right)\frac{\norm{f}_{L^2}}{2\sqrt{\pi} \mu^{1/4}},
\label{eq:estim_KSS_magnetic_2}
\end{equation}
as we wanted. The estimate for $2<p<\ii$ follows by complex interpolation~\cite{Simon01}. 
\end{proof}

We conclude the proof of Proposition~\ref{prop:growth} by recalling the (time-independent) gauge transformation $\boI_\theta$ 
\begin{equation}
\label{def:gauge-transform}
\boI_\theta(\psi) := e^{i \theta} \psi.
\end{equation}
This is a unitary operator which satisfies
$$\boI_\theta \mathscr{P}_A \boI_\theta^{- 1} = \mathscr{P}_{A+\nabla\theta} \quad {\rm and} \quad \boI_\theta G_\rho \boI_\theta^{- 1} = G_\rho.$$
Therefore,~\eqref{eq:gauge-invariance} follows from the invariance of the trace of the first term under conjugation by a unitary operator, and this concludes the proof of Proposition~\ref{prop:growth}.
\end{proof}

\subsubsection*{Step 5: Replacement by an almost constant potential}

After these preparations we are ready to start the proof of the semi-classical limit~\eqref{eq:conv-eps}. The first step is to replace $A_{\epsilon,y}$  by the potential $B(y)\times x/2$, up to a small error. We would like to express the error only in terms of the magnetic field $B$, up to a gradient term that will then be dropped using gauge invariance. The following simple formula is well known (see e.g.~\cite[Compl.~D$_{\rm IV}$]{CohDupGry-97}).

\begin{lemma}[Fundamental theorem of calculus in Poincar\'e gauge]
Let $A\in C^1(\R^3,\R^3)$ be any vector field. Then we have
\begin{equation}
 A(y+x)=\nabla_x\left(x\cdot \int_0^1 A(y+tx)\,dt\right)-x\times\int_0^1(\curl A)(y+tx)\,t\,dt.
 \label{eq:magnetic_Taylor}
\end{equation}
\end{lemma}

\begin{proof}
Using
$$\nabla \big(x\cdot a\big)=x\times (\curl a)+(x\cdot \nabla)a+a,$$
we can write the right side of~\eqref{eq:magnetic_Taylor} as
\begin{multline*}
\nabla_x\left(x\cdot \int_0^1 A(y+tx)\,dt\right)-x\times\curl_x \int_0^1A(y+tx)\,dt\\
= \int_0^1 \Big(\underbrace{t\,x\cdot \nabla A(y+tx)+ A(y+tx)}_{\frac{d}{dt}tA(y+tx)}\Big)\,dt= A(y+x).
\end{multline*}
\end{proof}

The idea of the decomposition~\eqref{eq:magnetic_Taylor} is that the first gradient term can be dropped using gauge invariance, leading to a new magnetic potential
$$\tilde A_y(x)= -x\times\int_0^1(\curl A)(y+tx)\,t\,dt\simeq B(y)\times x/2.$$
This potential does not belong to the Coulomb gauge anymore, but rather satisfies the Poincar\'e (also called multipolar and Fock-Schwinger) gauge condition at $y$
$$x\cdot A(x+y)=0.$$

Applying formula~\eqref{eq:magnetic_Taylor} in our situation gives
\begin{equation}
 \frac{A(y+\epsilon x)}\epsilon=\nabla_x\left(x\cdot \int_0^1 \frac{A(y+t\epsilon x)}\epsilon\,dt\right)+B(y)\times x/2+\epsilon\,R_{\epsilon,y}(x)
 \label{eq:magnetic_Taylor_with_epsilon}
\end{equation}
where
\begin{equation}
\boxed{R_{\epsilon,y}(x)=x\times\int_0^1\frac{B(y)-B(y+t\epsilon x)}{\epsilon}\,t\,dt}
\label{def:Repsy}
\end{equation}
is an error term.

The following gives some simple properties of the error term $R_{\epsilon,y}$.

\begin{lemma}[Estimates on $R_{\epsilon,y}$]\label{lem:estim_R_eps_y}
Let $\rho,\eps>0$, and $y\in\R^3$. Then we have for a universal constant $C$
\begin{multline}
\norm{R_{\epsilon,y}(x)}_{L^p_y(\R^3)} +|x|\norm{\curl R_{\epsilon,y}(x)}_{L^p_y(\R^3)}\\ \leq C|x|\,\min\left(|x|\norm{\nabla B}_{L^p(\R^3)},\frac{\norm{B}_{L^p(\R^3)}}{\eps}\right)
\end{multline}
and
\begin{equation}
 \norm{\div R_{\epsilon,y}(x)}_{L^p_y(\R^3)} \leq C\norm{\nabla B}_{L^p(\R^3)}\,|x|.
\end{equation}
\end{lemma}

\begin{proof}
We have
$$|R_{\epsilon,y}(x)|\leq \frac{|x|}{2\eps} \left(|B(y)|+2\int_0^1|B(y+t\eps x)|\,tdt\right).$$
Integrating over $y$ gives
$$\norm{R_{\epsilon,y}(x)}_{L^p_y(\R^3)}\leq \frac{|x|}{\eps} \norm{B}_{L^p(\R^3)}.$$
On the other hand, using
$$\frac{B(y)-B(y+t\epsilon x)}{\epsilon}=-t\sum_{j=1}^3\int_0^1 x_j(\partial_jB)(y+ts\epsilon x)\,ds,$$
and the identity
$$\int_0^1\int_0^1f(ts)\,t^2dt\,ds=\int_0^1\frac{1-t^2}{2}f(t)\,dt,$$
we can write
\begin{equation}
R_{\epsilon,y}(x)=-x\times \int_0^1\frac{1-t^2}{2}\sum_{j=1}^3 x_j(\partial_jB)(y+t\epsilon x)\,dt
\end{equation}
hence
$$\norm{R_{\epsilon,y}(x)}_{L^p(\R^3)}\leq C|x|^2\norm{\nabla B}_{L^p(\R^3)}.$$
Finally, we have
\begin{equation*}
 \curl R_{\epsilon,y}(x)=\frac{B(y+\epsilon x)-B(y)}{\epsilon},\quad \div R_{\epsilon,y}(x)=x\cdot \int_0^1\curl B(y+t\epsilon x)\,t^2\,dt
\end{equation*}
and the estimates are similar. 
\end{proof}

Since $B\in L^\ii(\R^3)$, we have $|\curl R_{\eps,y}|\leq C/\eps$ and we can apply Proposition~\ref{prop:growth}. 
The following is then an immediate consequence of~\eqref{eq:gauge-invariance} and~\eqref{eq:magnetic_Taylor_with_epsilon}.

\begin{cor}[Replacing by an almost constant field]
\label{lem:simplification}
Assume that $B$ satisfies the conditions~\eqref{eq:cond-B} of Theorem~\ref{thm:lim-eps}.  Then for any $\omega\in\R$ and any $y\in\R^3$, we have
 \begin{equation}
\label{eq:F-epsy}
f_{\omega}(e A_{\varepsilon,y}) =  f_{\omega}\left(\frac{e}{2} B(y) \times \cdot + e \varepsilon R_{\varepsilon, y}\right),
\end{equation}
where we recall that $R_{\varepsilon,y}$ is defined in~\eqref{def:Repsy}.
\end{cor}

\subsubsection*{Step 6: Computation for a constant field}
First we discard the error term $R_{\epsilon,y}$ in~\eqref{eq:F-epsy} and compute the exact value of the energy. Of course, we find the localized Pauli-Villars-regulated Euler-Heisenberg energy.

\begin{prop}[Constant field]\label{prop:value}
If $B$ is constant and $A(x)=B\times x/2$, then the function defined in~\eqref{eq:def_f_omega_A} equals
\begin{align}
f_{\omega}(B\times x/2)
=&\frac{1}{4\pi^{3/2}}  \int_0^\infty e^{-s\omega^2}\bigg( \sum_{j = 0}^2 c_j \, e^{- s m_j^2} \bigg) \Big( s |B| \coth \big( s |B| \big) - 1 \Big)\frac{ds}{s^{3/2}}.
\label{eq:computation_constant_field}
\end{align}
\end{prop}

Note that the value does not depend on the localization parameter $\rho>0$. From~\eqref{eq:computation_constant_field}, we obtain after integrating over $\omega$ that 
\begin{equation}
\frac1\pi\int_\R f_{\omega}(B\times x/2)\,\omega^2d\omega=f_{\rm vac}^{\rm PV}(|B|),
\end{equation}
the Pauli-Villars-regulated Euler-Heisenberg energy defined in~\eqref{def:boF-PV-EH}. The integral over $y$ then gives $\int f_{\rm vac}^{\rm PV}(|B(y)|)\,dy$ as we wanted.
The proof of Proposition~\ref{prop:value} is merely a computation, explained in Section~\ref{sec:comp-F*} below.

\subsubsection*{Step 7: Bound on the error}

At this step of the proof, we have shown in~\eqref{eq:integral_formula_resolvents_y},~\eqref{eq:F-epsy} and~\eqref{eq:computation_constant_field}  that
\begin{align*}
&\varepsilon^3\boF_{\rm vac}^{\rm PV}(A_\varepsilon)\\
&\qquad=\frac1\pi\int_{\R^3}dy\int_\R\omega^2\,d\omega\, f_\omega(eA_{\varepsilon,y})\\
&\qquad=\frac1\pi\int_{\R^3}dy\int_\R\omega^2\,d\omega\, f_{\omega}\left(\frac{e}{2} B(y) \times \cdot + e \varepsilon R_{\varepsilon, y}\right)\\
&\qquad=\int_{\R^3} f_{\rm vac}^{\rm PV}\big(|B(y)|\big)\,dy\\
&\qquad\quad+\frac1\pi\int_{\R^3}dy\int_\R\omega^2\,d\omega\, \left\{f_{\omega}\left(\frac{e}{2} B(y) \times \cdot + e \varepsilon R_{\varepsilon, y}\right)-f_{\omega}\left(\frac{e}{2} B(y) \times \cdot\right)\right\}
\end{align*}
and there only remains to evaluate the error term. The difficulty is, of course, to have an estimate on the integrand that can be integrated over $y$ and $\omega$.
This is the content of the following result.

\begin{prop}[Bound on the error term]\label{prop:bound_error}
Assume that $B$ satisfies the conditions of Theorem~\ref{thm:lim-eps}. Then we have for every $0<\epsilon\leq1$
\begin{equation}
 \int_{\R^3}dy \left|f_{\omega}\left(\frac{e}{2} B(y) \times \cdot + e \varepsilon R_{\varepsilon, y}\right)-f_{\omega}\left(\frac{e}{2} B(y) \times \cdot\right)\right|
 \leq C \frac{\varepsilon}{(m^2+\omega^2)^2},
\end{equation}
with a constant $C$ which only depends on $B$, on $\rho$ and on the $c_j$'s and $m_j$'s.
\end{prop}

Integrating over $\omega$ we find
$$ \int_\R\omega^2\,d\omega\,\int_{\R^3}dy \left|f_{\omega}\left(\frac{e}{2} B(y) \times \cdot + e \varepsilon R_{\varepsilon, y}\right)-f_{\omega}\left(\frac{e}{2} B(y) \times \cdot\right)\right|
 \leq C \varepsilon,$$
which ends the outline of the proof of Theorem~\ref{thm:lim-eps}. 

The proof of Proposition~\ref{prop:bound_error} is the most tedious part of the proof of our main result, and it is provided later in Section~\ref{sec:bound_error}. The next section is devoted to the proof of Proposition~\ref{prop:value}, whereas Appendix~\ref{sec:constant} gathers some important estimates for the resolvent of the Pauli operator with constant magnetic field.

\section{Computation for a constant field: proof of Proposition~\ref{prop:value}}
\label{sec:comp-F*}

Applying a rotation and using  the invariance of $-\Delta$ as well as the Gaussian $G_\rho$, we can always assume that $B$ is parallel to $e_3$, 
$$B=(0,0,b),$$
which we will do for the rest of the proof. 

Since we have shown in Proposition~\ref{prop:growth} that the two terms are separately trace-class, we can compute their trace separately.
The computation is easier if we go back to heat kernels using $H^{-1}=\int_0^\ii e^{-sH}\,ds$:
\begin{equation}
\sum_{j=0}^2c_j\frac{1}{H+m_j^2}= \int_0^\ii \bigg(\sum_{j=0}^2c_je^{-sm_j^2}\bigg)e^{-sH}\,ds.
\label{eq:integral_resolvent_heat}
\end{equation}
We recall from Lemma~\ref{lem:positivity} that $\sum_{j=0}^2c_je^{-sm_j^2}\geq0$. We conclude that
\begin{multline*}
\tr\left\{G_\rho \sum_{j=0}^2\frac{c_j}{H+\omega^2+m_j^2} G_\rho\right\}\\
=\int_0^\ii e^{-s\omega^2}\left(\sum_{j=0}^2c_je^{-sm_j^2}\right)\tr\big\{G_\rho e^{-sH}G_\rho\big\}\,ds\geq0
\end{multline*}
which will be used for $H=-\Delta$ and $H=\mathscr{P}_{B\times x/2}$.

For the Laplacian, we have
$$\tr G_\rho e^{s\Delta}G_\rho=2\left(\int G_\rho^2\right)\left(\frac{1}{(2\pi)^3}\int_{\R^3}e^{-sp^2}\,dp\right)=\frac{1}{4\pi^{3/2}s^{3/2}},$$
so 
$$\tr \left\{G_\rho \sum_{j=0}^2\frac{c_j}{-\Delta+\omega^2+m_j^2} G_\rho\right\}=\frac{1}{4\pi^{3/2}}\int_0^\ii e^{-s\omega^2}\left(\sum_{j=0}^2c_je^{-sm_j^2}\right)\frac{ds}{s^{3/2}}.$$

On the other hand, we recall in~\eqref{eq:formula_heat_kernel_constant_B} (Appendix~\ref{sec:constant}) that the heat kernel of the Pauli operator with constant magnetic field $B=b\,e_3$ is
\begin{multline}
\big[ e^{- s \mathscr{P}_{B\times x/2}} \big](x, y) = \frac{b}{8 \pi^{3/2} s^\frac{1}{2} \sinh(b s)} \times\\
\times e^{- \frac{b}{4} \coth(b s) |x^\perp - y^\perp|^2} \, e^{- \frac{1}{4 s} (x_3 - y_3)^2} \, e^{- \frac{i b}{2} x^\perp \times y^\perp} \, \begin{pmatrix} e^{b s} & 0 \\ 0 & e^{- b s} \end{pmatrix}
\label{eq:heat-Pauli1}
\end{multline}
where $x^\perp=(x_1,x_2)$ and $y^\perp=(y_1,y_2)$, which gives us
\begin{multline}
\label{eq:diago-Pauli2}
\tr_{\C^2}\big[ e^{- s \mathscr{P}_{B\times x/2}} \big](x, y) =  \frac{sb\coth(sb)}{4\pi^{3/2} s^\frac{3}{2} } \times\\
 \times e^{- \frac{b}{4} \coth(b s) |x^\perp - y^\perp|^2} \, e^{- \frac{1}{4 s} (x_3 - y_3)^2} \, e^{- \frac{i b}{2} x^\perp \times y^\perp}
\end{multline}
and therefore, the function being continuous at $x=y$, we deduce from~\cite[Thm.~2.12]{Simon01} that 
$$\tr G_\rho e^{- s \mathscr{P}_{B\times x/2}}G_\rho=\frac{sb\coth(sb)}{4\pi^{3/2} s^\frac{3}{2} },$$
which ends the proof of Proposition~\ref{prop:value}.\qed

\section{Bound on the error term: proof of Proposition~\ref{prop:bound_error}}
\label{sec:bound_error}

We follow here the same strategy as in the proof of Proposition~\ref{prop:growth}, with $\mathscr{P}_0=-\Delta$ replaced everywhere by the Pauli operator $\mathscr{P}_{e B(y) \times \cdot/2}$ with constant magnetic field $B(y)$. We will use pointwise estimates on the kernel of the resolvent of $\mathscr{P}_{e B(y) \times \cdot/2}$ which are recalled in Appendix~\ref{sec:constant}, that will essentially reduce the problem to the free case, up to a multiplicative constant. This method cannot be easily coupled to operators bounds, which forces us to use a 6th order expansion. In particular we will use that the operators appearing in the 6th order have a sign.

In order to simplify our notation, we introduce the Klein-Gordon operator with constant magnetic field $B(y)$ 
$$K_j(\omega,y):=\mathscr{P}_{e B(y) \times \cdot/2}+m_j^2+\omega^2.$$
Using the resolvent expansion as in~\eqref{eq:resolvent_expansion_free}, we have to estimate the trace of
\begin{align*}
T(\omega,y)&=G_\rho \sum_{j=0}^2c_j\left(\big(K_j(\omega,y)+S_{\eps,y}\big)^{-1}-K_j(\omega,y)^{-1}\right)G_\rho\\
&=G_\rho\sum_{n=1}^5(-1)^n\sum_{j=0}^2c_j\left(\frac{1}{K_j(\omega,y)}S_{\eps,y}\right)^n\frac{1}{K_j(\omega,y)}G_\rho\nonumber\\
&\qquad +G_\rho\sum_{j=0}^2c_j\left(\frac{1}{K_j(\omega,y)}S_{\eps,y}\right)^3\frac{1}{\mathscr{P}_{e B(y) \times \cdot/2+e\eps R_{\eps,y}}+m_j^2+\omega^2}\times\\
&\qquad\qquad\qquad\times\left(S_{\eps,y}\frac{1}{K_j(\omega,y)}\right)^3G_\rho\nonumber\\
&:=\sum_{n=1}^5T^{(n)}(\omega,y)+T^{(6)}(\omega,y),
\end{align*}
where the operator $S_{\eps,y}$ is defined by
\begin{align*}
S_{\eps,y} :=& \mathscr{P}_{e B(y) \times \cdot/2+e\eps R_{\eps,y}}-\mathscr{P}_{e B(y) \times \cdot/2}\\
=&e^2 \eps^2|R_{\eps,y}|^2-e \eps (\sigma\cdot p_y)(\sigma\cdot R_{\eps,y})-e \eps (\sigma\cdot R_{\eps,y})(\sigma\cdot p_y)\\
=&e^2 \eps^2|R_{\eps,y}|^2-e \eps (p_y\cdot R_{\eps,y}+R_{\eps,y}\cdot p_y)-e\eps \bssigma\cdot \curl R_{y,\eps},
\label{def:Syomega}
\end{align*}
with $p_y=p-eB(y)\times\cdot /2$ the magnetic momentum. The main result of this section is the estimate
\begin{equation}
\int_\R\int_{\R^3}\|T^{(n)}(\omega,y)\|_{\gS_1}\,dy\,\omega^2\,d\omega\leq C\eps^n,\qquad 1\leq n\leq 6,
\label{eq:to_be_proved}
\end{equation}
where the constant depends on $B$, on $\rho$ and on the $c_j$'s and $m_j$'s, which clearly ends the proof of Proposition~\ref{prop:bound_error}.

In order to prove~\eqref{eq:to_be_proved}, we will use pointwise kernel estimates, and in particular the fact that when $T$ and $T'$ are two trace-class operators such that 
$$|T(x,x')|\leq T'(x,x')$$
with the kernel of $T'$ being continuous in a neighborhood of the diagonal, then $|\tr(T)|\leq \tr(T')=\int_{\R^3}T'(x,x)\,dx$.

Our proof will rely on the following pointwise estimates on the kernels of the resolvents of the Pauli operator with constant magnetic field, proved in Proposition~\ref{prop:estimates_constant_B} in Appendix~\ref{sec:constant}:
\begin{equation}
\left|K_j(\omega,y)^{-1}(x,x')\right|\leq C(1+e\|B\|_{L^\ii})\,g_\omega(x-x')
\label{estim_kernel_K}
\end{equation}
and
\begin{equation}
\left|p_y\,K_j(\omega,y)^{-1}(x,x')\right|\leq C\frac{m_j}{m}(1+e^2\|B\|^2_{L^\ii})\,h_\omega(x-x'),
\label{estim_kernel_K_D}
\end{equation}
with the functions
\begin{equation}
 g_\omega(x)= \left( \frac{1}{|x|} + \frac{1}{\sqrt{m^2+\omega^2}} \right) e^{- \sqrt{m^2+\omega^2} |x|}
 \label{eq:def_f_omega}
\end{equation}
and
\begin{equation}
h_\omega(x)=\left(\frac{1}{|x |^2}+|x |^2+\frac1{m^2+\omega^2}+m^2+\omega^2 \right) e^{- \sqrt{m^2+\omega^2} |x|}.
\label{eq:def_g_omega}
\end{equation}
Note that
\begin{equation}
\norm{g_\omega}_{L^p(\R^3)}\leq \frac{C_p}{(m^2+\omega^2)^{\frac{3/p-1}{2}}},\qquad \norm{h_\omega}_{L^p(\R^3)}\leq \frac{C_p}{(m^2+\omega^2)^{\frac{3/p-2}{2}}},
\label{eq:estim_f_g_L_p}
\end{equation}
for all $1\leq p<\ii$. There are similar estimates for $|x|^\alpha g_\omega$ and $|x|^\alpha h_\omega$, with a better decay in $\omega$ when $\alpha>0$.

\subsubsection*{Sixth order term}
We start by estimating the kernel of the 6th order term $T^{(6)}(\omega,y)$. We are going to estimate the trace of each of the terms in the sum over $j=0,1,2$, which involves a non-negative operator and is thus always well defined.
We first use the operator bound 
\begin{equation*}
\frac{1}{\mathscr{P}_{e B(y) \times \cdot/2+\eps R_{\eps,y}}+m_j^2+\omega^2}\\ \leq \left(1+\frac{\mu}{m^2}\right)\frac{1}{\mathscr{P}_{e B(y) \times \cdot/2+\eps R_{\eps,y}}+m_j^2+\omega^2+\mu} 
\end{equation*}
as in the proof of Lemma~\ref{lem:KSS_magnetic}, with
$$\mu:=\|eB(y)+e\eps \curl R_{\eps,y}\|_{L^\ii(\R^3)}\leq 2e\|B\|_{L^\ii(\R^3)}$$
where the last estimate follows from Lemma~\ref{lem:estim_R_eps_y}. We obtain
\begin{multline*}
\int_\R \omega^2 d\omega \int_{\R^3} dy\,\|T^{(6)}(\omega,y)\|_{\gS_1}\\
\leq \left(1+\frac{2e\|B\|_{L^\ii(\R^3)}}{m^2}\right)\sum_{j=0}^2 |c_j| \int_\R \omega^2 d\omega \int_{\R^3} dy\;\tr\big\{T^{(6)}_j(\omega,y)\big\}
\end{multline*}
where
\begin{multline*}
T^{(6)}_j(\omega,y) =G_\rho\left(\frac{1}{K_j(\omega,y)}S_{\eps,y}\right)^3\times\\
\times\frac{1}{\mathscr{P}_{e B(y) \times \cdot/2+e\eps R_{\eps,y}}+m_j^2+\omega^2+\mu}\left(S_{\eps,y}\frac{1}{K_j(\omega,y)}\right)^3G_\rho \geq0.
\end{multline*}
We estimate the kernel of $T^{(6)}_j(\omega,y)$ by writing
$$S_{\eps,y}=e^2 \eps^2|R_{\eps,y}|^2-e \eps (2p_y\cdot R_{\eps,y}+i\div R_{\eps,y})-e\eps \bssigma\cdot \curl R_{y,\eps},$$
for the term on the left and using $2R_{\eps,y}\cdot p_y-i\div R_{\eps,y}$ for the term on the right. 
We also use the diamagnetic-type inequality~\eqref{eq:pointwise_bd_kernel_resolvents}
\begin{multline*}
\big|\big(\mathscr{P}_{e B(y) \times \cdot/2+e\eps R_{\eps,y}}+m_j^2+\omega^2+\mu\big)^{-1}(x,x')\big|_\ii\\ 
\leq \big|\big(-\Delta+m_j^2+\omega^2\big)^{-1}(x,x')\big|_\ii=\frac{e^{-\sqrt{m^2+\omega^2}|x-x'|}}{4\pi|x-x'|}\leq \frac{g_\omega(x-x')}{4\pi}
\end{multline*}
for the kernel of the middle operator, as well as $g_\omega\leq Ch_\omega$ to simplify our bounds. We find
\begin{multline*}
|T^{(6)}_j(\omega,y)(x,x')|\leq C\,G_\rho(x) \int_{\R^3}dz_1\cdots\int_{\R^3}dz_6\\
h_\omega(x-z_1)\prod_{k=1}^5 h_{\eps,y}(z_k)h_\omega(z_k-z_{k+1})h_{\eps,y}(z_6)h_\omega(z_6-x')G_\rho(x')
\end{multline*}
with
\begin{align*}
h_{\eps,y}(z)&:=e\eps|R_{\eps,y}(z)|+e\eps|\div R_{\eps,y}(z)|+e\eps|\curl R_{\eps,y}(z)|+e^2\eps^2|R_{\eps,y}(z)|^2\\
&\leq e\eps\big(1+e\|B\|_{L^\ii(\R^3)}|z|\big)|R_{\eps,y}(z)|+e\eps|\div R_{\eps,y}(z)|+e\eps|\curl R_{\eps,y}(z)|.
\end{align*}
Using Lemma~\ref{lem:estim_R_eps_y} together with the assumption that $B,\nabla B\in L^p(\R^3)$ for $1\leq p\leq6$, we obtain
\begin{equation}
\forall 1\leq p\leq 6,\qquad \norm{h_{\eps,y}(z)}_{L^p_y(\R^3)}\leq C\eps(1+|z|^3).
\label{eq:estim_h}
\end{equation}
Next we integrate over $y$, use H\"older's inequality and the bound~\eqref{eq:estim_h} for $p=6$. We arrive at
\begin{multline*}
\int_{\R^3}|T^{(6)}_j(\omega,y)(x,x')|\,dy\leq C\eps^6\,G_\rho(x) \int_{\R^3}dz_1\cdots\int_{\R^3}dz_6\,h_\omega(x-z_1)\times\\
\times\prod_{k=1}^5 (1+|z_k|^3)h_\omega(z_k-z_{k+1})(1+|z_6|^3)h_\omega(z_6-x')G_\rho(x').
\end{multline*}
We can put the polynomial terms together with the functions $h_\omega$ using that 
\begin{equation}
\prod_{k=1}^J (1+|z_k|^3)\leq C_J\prod_{k=1}^{J-1} (1+|z_k-z_{k+1}|^3)(1+|z_J-x'|^3)(1+|x-x'|^3)(1+|x-z_1|^3)
\end{equation}
and
$$(1+|x-x'|^3)G_\rho(x) G_\rho(x')=(\pi\rho)^{-3}(1+|x-x'|^3)e^{-\frac{|x+x'|^2+|x-x'|^2}{2\rho^2}}\leq C e^{-\frac{|x+x'|^2}{2\rho^2}}.$$
The above bound becomes
\begin{equation*}
\int_{\R^3}|T^{(6)}_j(\omega,y)(x,x')|\,dy\leq C\eps^6\,e^{-\frac{|x+x'|^2}{2\rho^2}} \big\{(1+|\cdot|^3)h_\omega\big\}^{\ast 7}(x-x')
\end{equation*}
where $f^{\ast n}$ is the $n$-fold convolution of a function $f$. By Young's inequality and~\eqref{eq:estim_f_g_L_p} we find
\begin{align*}
\int_{\R^3}\tr T^{(6)}_j(\omega,y)\,dy&\leq C\eps^6 \norm{(1+|\cdot|^3)h_\omega}_{L^2(\R^3)}^2\norm{(1+|\cdot|^3)h_\omega}^5_{L^1(\R^3)}\\
&\leq\frac{C\eps^6}{(m^2+\omega^2)^2}.
\end{align*}
Suming over $j$, we have proved that
$$\int_\R\int_{\R^3}\|T^{(6)}(\omega,y)\|_{\gS_1}\,dy\,\omega^2\,d\omega\leq C\eps^6,$$
as we wanted, where the constant depends on $B$, on $\rho$ and on the $c_j$'s and $m_j$'s.

\medskip

The proof for the other terms is very similar, and relies on the Pauli-Villars condition~\eqref{eq:cond-PV}, exactly as in the proof of Proposition~\ref{prop:extract_non_trace_class}. We only sketch it.

\subsubsection*{First order term}
We write
\begin{align*}
-T^{(1)}(\omega,y)&=G_\rho\bigg(\sum_{j=1}^2c_j(m_j^2-m^2)^2\frac{1}{K_j(\omega,y)K_0(\omega,y)}S_{\eps,y}\frac{1}{K_j(\omega,y)K_0(\omega,y)}\nonumber\\
&\qquad+\sum_{j=1}^2c_j(m_j^2-m^2)^2\frac{1}{K_0(\omega,y)}S_{\eps,y}\frac{1}{K_j(\omega,y)K_0(\omega,y)^2}\nonumber\\
&\qquad+\sum_{j=1}^2c_j(m_j^2-m^2)^2\frac{1}{K_j(\omega,y)K_0(\omega,y)^2}S_{\eps,y}\frac{1}{K_0(\omega,y)}\bigg) G_\rho.\label{eq:annulations_bis}
\end{align*}
Since $S_{\eps,y}$ contains only one operator $p_y$, estimating the kernel of the above operators will involve one function $h_\omega$ and three functions $g_\omega$. The argument is exactly the same as before, using this time $\norm{h_{\eps,y}(z)}_{L^1_y(\R^3)}$, with the final bound
$$\int_{\R^3}|T^{(1)}(\omega,y)(x,x')|\,dy\leq C\,\eps\,e^{-\frac{|x+x'|^2}{2\rho^2}} \{(1+|\cdot|^3)h_\omega\}\ast\{(1+|\cdot|^3)g_\omega\}^{\ast 3}(x-x')$$
and thus
\begin{align*}
&\int_{\R^3}\big|\tr T^{(1)}(\omega,y)\big|\,dy\\
&\qquad\leq C\eps\norm{(1+|\cdot|^3)h_\omega}_{L^1(\R^3)}\norm{(1+|\cdot|^3)g_\omega}_{L^1(\R^3)}\norm{(1+|\cdot|^3)g_\omega}^2_{L^2(\R^3)}\\
&\qquad\leq C\frac{\eps}{(m^2+\omega^2)^2}
\end{align*}
which proves~\eqref{eq:to_be_proved} for $n=1$.

\subsubsection*{Second order term}
For the second order term, we use as in~\eqref{eq:dev-gT2}
\begin{equation}
\label{eq:dev-gT2y}
\begin{split}
 T^{(2)}(\omega,y) =& \sum_{j= 0}^2 c_j \, (m_j^2 - m_0^2)^2 \times\\
 & \times G_\rho\Big( \frac{1}{K_j(\omega,y)K_0(\omega,y)^2} \, S_{\eps,y} \, \frac{1}{K_j(\omega,y)} \, S_{\eps,y} \, \frac{1}{K_j(\omega,y)}\\
& + \frac{1}{K_0(\omega,y)^2} \, S_{\eps,y} \, \frac{1}{K_j(\omega,y)K_0(\omega,y)} \, S_{\eps,y} \, \frac{1}{K_j(\omega,y)}\\
& + \frac{1}{K_0(\omega,y)^2} \, S_{\eps,y} \, \frac{1}{K_0(\omega,y)} \, S_{\eps,y} \, \frac{1}{K_j(\omega,y)K_0(\omega,y)}\\
& + \frac{1}{K_0(\omega,y)} \, S_{\eps,y} \, \frac{1}{K_j(\omega,y)K_0(\omega,y)^2} \, S_{\eps,y} \, \frac{1}{K_j(\omega,y)}\\
&+ \frac{1}{K_0(\omega,y)} \, S_{\eps,y} \, \frac{1}{K_0(\omega,y)^2} \, S_{\eps,y} \, \frac{1}{K_j(\omega,y)K_0(\omega,y)}\\
& + \frac{1}{K_0(\omega,y)} \, S_{\eps,y} \, \frac{1}{K_0(\omega,y)} \, S_{\eps,y} \, \frac{1}{K_j(\omega,y)K_0(\omega,y)^2} \Big)G_\rho
\end{split}
\end{equation}
and apply the same argument. There are two functions $h_\omega$, three functions $g_\omega$, and the $L^2_y$ norm of $h_{\eps,y}$. We get
\begin{align*}
&\int_{\R^3}\big|\tr T^{(2)}(\omega,y)\big|\,dy\\
&\qquad\leq C\eps^2\norm{(1+|\cdot|^3)h_\omega}^2_{L^1(\R^3)}\norm{(1+|\cdot|^3)g_\omega}_{L^1(\R^3)}\norm{(1+|\cdot|^3)g_\omega}^2_{L^2(\R^3)}\\
&\qquad\leq C\frac{\eps^2}{(m^2+\omega^2)^{\frac52}}.
\end{align*}

\subsubsection*{Third, fourth and fifth order terms}
For the other terms we only use the first Pauli-Villars condition, in order to simplify the argument. For instance, we write
\begin{align*}
-T^{(3)}(\omega,y)=&G_\rho\sum_{j=0}^2c_j\left(\frac{1}{K_j(\omega,y)}S_{\eps,y}\right)^3\frac{1}{K_j(\omega,y)}G_\rho\\
=&G_\rho\sum_{j=0}^2c_j(m^2-m_j^2)\bigg\{\frac{1}{K_0(\omega,y)K_j(\omega,y)}\left(S_{\eps,y}\frac{1}{K_j(\omega,y)}\right)^3\\
&+\frac{1}{K_0(\omega,y)}S_{\eps,y}\frac{1}{K_0(\omega,y)K_j(\omega,y)}\left(S_{\eps,y}\frac{1}{K_j(\omega,y)}\right)^2\\
&+\left(\frac{1}{K_0(\omega,y)}S_{\eps,y}\right)^2\frac{1}{K_0(\omega,y)K_j(\omega,y)}S_{\eps,y}\frac{1}{K_j(\omega,y)}\\
&+\left(\frac{1}{K_0(\omega,y)}S_{\eps,y}\right)^3\frac{1}{K_0(\omega,y)K_j(\omega,y)}\bigg\}G_\rho
\end{align*}
and estimate its kernel as before. There are three functions $h_\omega$, two functions $g_\omega$, and the $L^3_y$ norm of $h_{\eps,y}$. We get
\begin{multline*}
\int_{\R^3}\big|\tr T^{(3)}(\omega,y)\big|\,dy\\
\leq C\eps^3\norm{(1+|\cdot|^3)h_\omega}^3_{L^1(\R^3)}\norm{(1+|\cdot|^3)g_\omega}^2_{L^2(\R^3)}\leq C\frac{\eps^3}{(m^2+\omega^2)^2}.
\end{multline*}
The argument for the fourth and fifth order term is exactly the same. This ends the proof of Proposition~\ref{prop:bound_error}.\qed

\appendix
\section{The Pauli operator with constant magnetic field}
\label{sec:constant}

In this appendix, we gather some useful properties of the Pauli operator with constant magnetic field $B\in\R^3$
\begin{equation}
\mathscr{P}_{B\times\cdot/2}=\big( \bssigma \cdot (- i \nabla -  B\times x/2) \big)^2.
\label{eq:def_Pauli_constant_B}
\end{equation}
The Pauli operator naturally splits into a two-dimensional Pauli operator in the plane orthogonal to $B$ and a Laplacian in the direction of $B$. Here as in the sequel, we denote by $x^\perp$ the projection of any vector $x\in\R^3$ on this plane. 

In the text we need the following pointwise estimates on the kernel of the resolvent of $\mathscr{P}_{B\times\cdot/2}$ as well as on $(- i \partial_j -  (B\times x)_j/2)$ times the resolvent. 
\begin{prop}[Pointwise estimates for a constant magnetic field]\label{prop:estimates_constant_B}
Let $B\in\R^3$ be a constant magnetic field. Then we have the pointwise estimates
\begin{equation}
\label{eq:estim_resolvent_constant_B}
\Big|\big(\mathscr{P}_{B\times\cdot/2}+\mu^2\big)^{-1}(x,y)\Big|_\ii\leq \frac{1}{4\pi}\Big( \frac{1}{|x - y|} + \frac{2|B|}{ \mu} \Big) e^{- \mu |x - y|}
\end{equation}
and
\begin{multline}
\label{eq:estim_resolvent_constant_B_derivative}
\Big|(-i\partial_j-(B\times x)_j/2)\big(\mathscr{P}_{B\times\cdot/2}+\mu^2\big)^{-1}(x,y)\Big|_\ii\\
\leq \Big( \frac{|B|}{2 \pi} + \frac{2}{\pi |x - y|^2} + \frac{2 \mu}{\pi |x - y|} + \frac{5 |B|^2}{4 \pi \mu} |x^\perp - y^\perp| \Big) e^{- \mu |x - y|}
\end{multline}
for every $\mu>0$ and $j=1,2,3$, where $|M|_\ii=\sup\|x\|^{-1}\|Mx\|$ is the sup norm of $2\times2$ matrices.
\end{prop}

Similar estimates were derived in~\cite[Lemma A.10]{ErdoSol3}. We have already shown a pointwise bound in Lemma~\ref{lem:KSS_magnetic} for the resolvent~\eqref{eq:estim_resolvent_constant_B}, but the estimate~\eqref{eq:estim_resolvent_constant_B_derivative} requires a bit more work.

\begin{proof}
For the proof of the proposition we can assume, after applying a suitable rotation, that $B=|B|e_3$ and then rewrite
$$\mathscr{P}_{B\times\cdot/2}=-(\partial_1+i|B|x_2/2)^2-(\partial_2-i|B|x_1/2)^2-\partial_3^2-|B|\sigma_3.$$
Next we express the resolvent using the heat kernel as follows
$$\big(\mathscr{P}_{B\times\cdot/2}+\mu^2\big)^{-1}(x,y)=\int_0^\ii e^{-s\mathscr{P}_{B\times\cdot/2}}(x,y)e^{-s\mu^2}\,ds.$$
Note that the operators $|B|\sigma_3$ and $\partial_3^2$ commute with each other, as well as with the first two terms. Hence
\begin{multline*}
 e^{-s\mathscr{P}_{B\times\cdot/2}}(x,y)\\
 =e^{s(\partial_1+i|B|x_2/2)^2+s(\partial_2-i|B|x_1/2)^2}(x^\perp,y^\perp)\frac{e^{-\frac{(x_3-y_3)^2}{4s}}}{2\sqrt{\pi s}}\begin{pmatrix}
e^{|B|s}&0\\
0&e^{-|B|s}
\end{pmatrix}
\end{multline*}
with $x^\perp=(x_1,x_2)$. The kernel of the two-dimensional Pauli operator is
\begin{multline*}
e^{s(\partial_1+i|B|x_2/2)^2+s(\partial_2-i|B|x_1/2)^2}(x^\perp,y^\perp)\\
=\frac{|B|}{4 \pi \sinh(|B| s)} \, e^{- \frac{|B|}{4} \coth(|B| s) |x^\perp - y^\perp|^2} \, e^{- \frac{i |B|}{2} x^\perp \times y^\perp}
\end{multline*}
(see e.g.~\cite[Chapter 15]{Simon02}). As a consequence, we find
\begin{multline}
  e^{-s\mathscr{P}_{B\times\cdot/2}}(x,y)
 =\frac{|B|}{8 \pi^{3/2} \sqrt{s} \sinh(|B| s)}e^{- \frac{|B|}{4} \coth(|B| s) |x^\perp - y^\perp|^2} \times \\ \times  e^{- \frac{i |B|}{2} x^\perp \times y^\perp}e^{-\frac{(x_3-y_3)^2}{4s}}\begin{pmatrix}
e^{|B|s}&0\\
0&e^{-|B|s}
\end{pmatrix}
\label{eq:formula_heat_kernel_constant_B}
\end{multline}
and
\begin{multline}
\big(\mathscr{P}_{B\times\cdot/2}+\mu^2\big)^{-1}(x,y)=\frac{e^{- \frac{i |B|}{2} x^\perp \times y^\perp}}{8 \pi^{3/2}}\int_0^\ii \frac{|B|}{\sqrt{s} \sinh(|B| s)} \begin{pmatrix}
e^{|B|s}&0\\
0&e^{-|B|s}
\end{pmatrix}\times\\
\times e^{- \frac{|B|}{4} \coth(|B| s) |x^\perp - y^\perp|^2} e^{-\frac{(x_3-y_3)^2}{4s}}e^{-s\mu^2}\,ds. 
\label{eq:ker-Res}
\end{multline}

In order to bound this quantity, we split the domain of integration into two pieces. When $|B| s \leq 1$, we use the inequalities $|B| s \leq \sinh(|B| s)$, $|B| s \coth(|B|s) \geq 1$ and $\exp(|B|s)\leq 1+(e-1)s|B|\leq 1+4s|B|$. 
When $|B| s \geq 1$, we also use the inequality $|B| s \coth(|B|s) \geq 1$, as well as the bound $\exp(|B| s) \leq 4 \sinh(|B| s)$. This gives
\begin{multline*}
\left|\big(\mathscr{P}_{B\times\cdot/2}+\mu^2\big)^{-1}(x,y)\right|_\ii\\ 
\leq \int_0^{\frac{1}{|B|}} \frac{ds}{(4 \pi s)^\frac{3}{2}} e^{- s \mu^2} e^{- \frac{|x - y|^2}{4 s}}+\frac{|B|}{\pi} \int_0^\infty \frac{ds}{\sqrt{4 \pi s}} e^{- s \mu^2} e^{- \frac{|x - y|^2}{4 t}}. 
\end{multline*}
It then remains to recall the following formulae for the kernels of the resolvent of the one-dimensional and three-dimensional Laplace operators
\begin{equation}
\label{eq:Res-Laplace-1d}
\int_0^\infty \frac{ds}{\sqrt{4 \pi s}} e^{- \mu^2 s} e^{- \frac{r^2}{4 s}} = \frac{e^{- \mu r}}{2 \mu},
\end{equation}
respectively,
\begin{equation}
\label{eq:Res-Laplace-3d}
\int_0^\infty \frac{ds}{(4 \pi s)^\frac{3}{2}} e^{- \mu^2s} e^{- \frac{r^2}{4 s}} = \frac{e^{- \mu r}}{4 \pi r},
\end{equation}
in order to obtain the estimate~\eqref{eq:estim_resolvent_constant_B}.

Concerning~\eqref{eq:estim_resolvent_constant_B_derivative}, we first derive from~\eqref{eq:ker-Res} that 
\begin{align*}
&\big[ (-i\partial_1+|B|x_2)  \big(\mathscr{P}_{B\times\cdot/2}+\mu^2\big)^{-1}\big](x,y) \\
&\qquad= \frac{|B|^2}{2 (4 \pi)^\frac{3}{2}} e^{- \frac{i |B|}{2} x^\perp \times y^\perp} \int_0^\infty e^{- \frac{|B|}{4} \coth(|B| s) |x^\perp - y^\perp|^2} \, e^{- \frac{(x_3 - y_3)^2}{4 s} } \times\\ 
&\qquad\qquad \times \begin{pmatrix} e^{|B| s} & 0 \\ 0 & e^{- |B|s} \end{pmatrix}\big( x_2 - y_2 + i (x_1 - y_1) \coth(|B| s) \big) \, \frac{e^{- \mu^2 s} \, ds}{s^\frac{1}{2} \sinh(|B|s)}.
\end{align*}
Arguing as in the proof of~\eqref{eq:estim_resolvent_constant_B}, and applying the inequalities $|B| s \exp(|B| s) \linebreak[0] \coth(|B| s) \leq 8$ for $|B| s \leq 1$, respectively $\coth(|B| s) \leq 4$ for $|B| s \geq 1$, we are led to the estimate
\begin{align*}
&\left|\big[ (-i\partial_1+|B|x_2)  \big(\mathscr{P}_{B\times\cdot/2}+\mu^2\big)^{-1}\big](x,y)\right|_\ii\\
&\qquad\leq |x^\perp - y ^\perp| \bigg( 2 |B| \int_0^\infty \frac{ds}{(4 \pi s)^\frac{3}{2}} e^{- \mu^2 s} e^{- \frac{|x - y|^2}{4 s}}\\
& \qquad\qquad+ 16 \pi \int_0^\infty \frac{ds}{(4 \pi s)^\frac{5}{2}} e^{- \mu^2 s} e^{- \frac{|x - y|^2}{4 s}} + \frac{5 |B|^2}{2 \pi} \int_0^\infty \frac{ds}{\sqrt{4 \pi s}} e^{- \mu^2 s} e^{- \frac{|x - y|^2}{4 s}} \bigg).
\end{align*}
Inequality~\eqref{eq:estim_resolvent_constant_B_derivative} then follows from~\eqref{eq:Res-Laplace-1d},~\eqref{eq:Res-Laplace-3d}, as well as the identity
$$\int_0^\infty \frac{ds}{(4 \pi s)^\frac{5}{2}} e^{- \mu^2s} e^{- \frac{r^2}{4 s}} = \frac{e^{- \mu r}}{8 \pi^2 r^2} \Big( \frac{1}{r} + \mu \Big),$$
which can be derived from~\eqref{eq:Res-Laplace-1d} and~\eqref{eq:Res-Laplace-3d} by integrating by parts.

Finally, the kernels of the two other operators are given by
\begin{align*}
&\big[ (-i\partial_2-|B|x_1)  \big(\mathscr{P}_{B\times\cdot/2}+\mu^2\big)^{-1}\big](x,y) \\
&\qquad= \frac{|B|^2}{2 (4 \pi)^\frac{3}{2}} e^{- \frac{i |B|}{2} x^\perp \times y^\perp} \int_0^\infty e^{- \frac{|B|}{4} \coth(|B| s) |x^\perp - y^\perp|^2} \, e^{- \frac{(x_3 - y_3)^2}{4 s} } \times\\ 
&\qquad\qquad \times \begin{pmatrix} e^{|B| s} & 0 \\ 0 & e^{- |B|s} \end{pmatrix}\big( y_1 - x_1 + i (x_2 - y_2) \coth(b t)\big) \, \frac{e^{- \mu^2 s} \, ds}{s^\frac{1}{2} \sinh(|B|s)}.
\end{align*}
and
\begin{align*}
&\big[ (-i\partial_3)  \big(\mathscr{P}_{B\times\cdot/2}+\mu^2\big)^{-1}\big](x,y) \\
&\qquad= (y_3-x_3)\frac{e^{- \frac{i |B|}{2} x^\perp \times y^\perp}}{2 (4 \pi)^\frac{3}{2}}\int_0^\ii e^{- \frac{|B|}{4} \coth(|B| s) |x^\perp - y^\perp|^2} e^{-\frac{(x_3-y_3)^2}{4s}}\times\\
&\qquad\qquad \times  \begin{pmatrix}
e^{|B|s}&0\\
0&e^{-|B|s}
\end{pmatrix}\frac{|B|e^{-s\mu^2}\,ds}{s^{\frac32} \sinh(|B| s)}.
\end{align*}
and the proof of~\eqref{eq:estim_resolvent_constant_B_derivative} for $j=2,3$ is similar.
This completes the proof of Proposition~\ref{prop:estimates_constant_B}.
\end{proof}

\begin{merci}
The authors thank J\"urg Fr\"ohlich for mentioning the work of Haba. M.L. and \'E.S. acknowledge support from the French Ministry of Research (Grant ANR-10-0101). M.L. acknowledges support from the European Research Council under the European Community's Seventh Framework Programme (FP7/2007-2013 Grant Agreement MNIQS 258023).
\end{merci}


\end{document}